\documentclass{article}
\pdfoutput=1
\usepackage[utf8]{inputenc} 
\usepackage[OT4]{fontenc} 
\usepackage{amsmath}
\usepackage{amssymb}
\usepackage{amsthm} 
\usepackage{mathtools}
\usepackage{hyperref}
\usepackage{subdepth} 
\usepackage{fullpage}
\usepackage{float} 
\usepackage[shortlabels]{enumitem} 
\usepackage{tikz}
\usepackage{tikz-3dplot}
\usepackage{multicol}
\usetikzlibrary{calc, positioning, arrows.meta, shapes.geometric, backgrounds, angles, quotes}
\usepackage{subcaption}
\usepackage{caption}
\usepackage{multirow}
\usepackage{doi}
\usepackage{cellspace} 
\usepackage{booktabs}
\usepackage{needspace}
\DeclareMathOperator{\conv}{conv}
\DeclareMathOperator{\tr}{tr}
\DeclareMathOperator{\aff}{aff}
\DeclareMathOperator{\interior}{int}
\DeclareMathOperator{\lin}{lin}
\DeclareMathOperator{\sgn}{sgn}

\theoremstyle{plain}
\newtheorem{theorem}{Theorem}
\newtheorem{corollary}[theorem]{Corollary}

\newtheorem{proposition}[theorem]{Proposition}
\newtheorem{lemma}[theorem]{Lemma}
\theoremstyle{remark}
\newtheorem{remark}{Remark}
\theoremstyle{definition}
\newtheorem{definition}[theorem]{Definition}

\makeatletter
\def\th@plain{
  \normalfont 
}
\makeatother

\title{From minimal informationally complete measurements to orthocentric simplices and back again}
\author{Piotr Bereza, Wojciech S\l{}omczy\'{n}ski, Anna Szymusiak\\[4pt]
\normalsize Institute of Mathematics, Jagiellonian University,\\
\normalsize \L{}ojasiewicza 6, 30-348 Krak\'{o}w, Poland\\[4pt]
\normalsize\href{mailto:piotr.bereza@doctoral.uj.edu.pl}{\texttt{piotr.bereza@doctoral.uj.edu.pl}},
\href{mailto:wojciech.slomczynski@uj.edu.pl}{\texttt{wojciech.slomczynski@uj.edu.pl}},
\href{mailto:anna.szymusiak@uj.edu.pl}{\texttt{anna.szymusiak@uj.edu.pl}}}
\date{\today}

\begin{document}
\maketitle

\begin{abstract}
The reconstruction of unknown quantum states via minimal informationally complete measurements (MICs) is a cornerstone of quantum tomography.
Although the statistical properties of these measurements are well-understood, their geometric structure has remained elusive.
In this work, we establish a correspondence between the class of minimal $s$-tight informationally complete measurements, encompassing, among others, tight IC and morphophoric measurements, and the classical geometry of orthocentric simplices.
In particular, we prove a three-way equivalence: a MIC is $s$-tight if and only if its measurement vectors, upon suitable rescaling, form the vertices of an acute orthocentric simplex with the orthocentre at the origin, and such simplices are precisely the homothetically self-dual ones.
This geometric manifestation of operational ``tightness'' provides a bridge between the physical world and Euclidean geometry.
Furthermore, the $s$-tight class is fully characterised by its measurement directions: the angles between them must be obtuse and satisfy a cross-ratio condition.
We determine the space of admissible direction configurations: modulo rotations, every such configuration is encoded by a single probability vector, the \emph{skeleton} of the measurement, together with an orientation class, so that the moduli space of $s$-tight MIC directions is $\Delta^{\circ}_{d+1}\times\{\pm 1\}$.
Conversely, every acute orthocentric simplex with the orthocentre at the origin can be anchored in the state space, generating a class of minimal $s$-tight IC measurements that contains exactly one tight IC measurement up to overall rescaling.
Although our primary formalism is constructed for geometric generalised
probabilistic theories, the resulting classification yields new insights
into the geometry of quantum measurements and provides a geometric toolkit for the design and analysis of minimal tomographic procedures in
both quantum and beyond-quantum theories.
\end{abstract}

\section{Introduction}

The search for the fundamental principles governing quantum theory has led to the development of \textit{generalised probabilistic theories} (\textit{GPTs}) -- a broad operational framework that treats quantum mechanics as one of many possible statistical physical theories~\cite{masanes2011derivation,janotta2014generalized,muller2021probabilistic,takakura2022convexity,plavala2023general}.
While standard approaches formulate these theories using ordered vector spaces, equipping the state space with an inner product -- or analogous geometric structures like angles and orthogonality -- naturally yields the concept of a \textit{geometric GPT} (\textit{GGPT}), as developed in~\cite{szymusiak2025can}.
This perspective distinguishes between the \textit{internal geometry} of the state space -- its intrinsic metric properties -- and its \textit{external geometry}, which describes how the state space is embedded within a larger vector space.
In this work, we propose a novel approach to GGPTs that focuses exclusively on the internal geometry. Here, states are represented as elements of a convex subset of a finite-dimensional vector space, while effects are described as pairs consisting of a positive real scalar (between $0$ and $1$) and a vector from the dual convex set.
This framework can be viewed as a generalisation of the \textit{Bloch representation} in quantum theory~\cite{bloch1946nuclear,kimura2003bloch,scott2006tight,appleby2007symmetric,bertlmann2008bloch,bengtsson2017geometry,eltschka2021shape}, extended to incorporate effects, measurements, and, eventually, instruments. 

To investigate this internal geometry, we require only three elements: a convex body of \textit{states} $S$ in a finite-dimensional real inner product space, a \textit{distinguished state}  $0$, and \textit{measurements} represented by outcome probabilities $c$ and vectors $\Psi$ from the dual set $S^\star$ that satisfy a closure condition. Formally:

\begin{itemize} 
    \item a convex and compact set (convex body) of states $S \subset V_0$, where $(V_0, \langle \cdot, \cdot \rangle)$ is a real inner product vector space;
    \item a distinguished state $0 \in \interior S$; 
    \item measurements, represented by pairs $(\Psi, c)$, where $c=(c_j)_{j=1}^n$ is a vector of non-zero \textit{measurement probabilities} (i.e., $\sum_{j=1}^n c_j=1$), and $\Psi=(\psi_j)_{j=1}^n$ is a sequence of \textit{measurement vectors} satisfying two conditions:
    \begin{itemize} 
        \item $\Psi \subset S^\star \coloneqq \{x \in V_0 : \langle x,y \rangle \geq -1 \quad \text{for all } y \in S\}$; 
        \item $\sum_{j=1}^n c_j \psi_j = 0$ (the \textit{closure condition}).
    \end{itemize}
\end{itemize}

Such a pair $(\Psi,c)$ generates the probability $\pi_j(x)$ of obtaining the measurement outcome $j \in \{1, \dots, n\}$ when the system is prepared in the state $x \in S$. This is given by the affine formula
\begin{equation} \label{eq:pj}
\pi_j(x) = c_j(\langle x, \psi_j \rangle + 1).
\end{equation}
This relation follows directly from the fact that probability $\pi$ is an affine function of the state, i.e., from the \textit{law of total probability}. In particular, $\pi(0)=c$. If $c = (1/n,\dots,1/n)$, we call the measurement \textit{unbiased}.

It should be noted that $S^\star$ is always closed and convex, while the condition $0 \in \interior S$ is equivalent to the boundedness of $S^\star$ (and thus its compactness).
Moreover, this framework allows us to distinguish between three types of GGPTs, depending on the relative positioning of the state set and its dual: \textit{infra-dual} ($S \subset S^\star$), \textit{supra-dual} ($S^\star \subset S$), and, most significantly, \textit{self-dual} ($S^\star = S$), the latter of which encompasses both the \textit{classical} and \textit{quantum} cases~\cite{bengtsson2013geometry}.

The measurement vectors $\psi_j$ admit a direct physical interpretation as \textit{post-measurement states}.
Making this precise requires the notion of a \textit{measurement instrument}, and it is this interpretation that will ultimately motivate the classes of measurements studied in this paper.
In this approach, an \textit{instrument} $\mathcal{I} \coloneqq ({\mathcal{I}}_j)_{j=1}^n$ is generated by a sequence of pairs $(F_j, \pi_j)_{j=1}^n$, where $\pi_j \colon S \to \mathbb{R}^+$, $F_j \colon S_j = \{x \in S \colon \pi_j(x) \neq 0\} \to S$ for $j = 1,\dots,n$, and $\sum_{j=1}^n \pi_j = 1$, where both $\pi_j$ and $\mathcal{I}_j \coloneqq \pi_j F_j$ are assumed to be affine. 
If the functions $\pi_j$, $j=1,\dots,n$, are defined by eq.~\eqref{eq:pj}, which implies $\pi_j(0) = c_j$, then we say that the \textit{instrument} $\mathcal{I} \coloneqq ({\mathcal{I}}_j)_{j=1}^n$ is an \textit{instrument for the measurement} $(\Psi,c)$. 
Let us interpret the notion of an instrument in the language of GGPT. If $x \in S$ represents the state of the system before the measurement, $\pi_j(x)$ denotes the probability that the measurement outcome equals $j$, and $F_j(x)$ represents the state of the system after the measurement, assuming that the outcome was actually $j$, then the fact that $\pi_j$ is affine is equivalent to the \textit{law of total probability}, and the property that $\mathcal{I}_j = \pi_j F_j$ is affine is equivalent to the (classical) \textit{Bayes' rule} (see \cite[Sec. 6]{slomczynski2003dynamical}, \cite[Sec. 2.4.3]{breuer2002theory}, and Appendix~\ref{app:instrument}). This approach was considered earlier in
\cite[Definition 6.1]{slomczynski2003dynamical}, see also~\cite{szczepanek2019dynamical}, where a sequence of pairs $(F_j, \pi_j)_{j=1}^n$ was called \textit{homogeneous partial iterated function system}.

Now, if we assume that the GGPT is supra-dual ($S^\star \subset S$) and an instrument $\mathcal{I}$ is \textit{balanced at} $0$, i.e., $F_j(0)=\psi_j$ for every $j=1,\dots,n$ (see \cite[Sec. 6.1]{szymusiak2025can}), the measurement vectors $\psi_j \in S$ can acquire a direct physical interpretation as \textit{post-measurement states} provided that the system was initially in state $0$.
The assumption that $\mathcal{I}$ is balanced at $0$ is highly natural in the quantum case, for instance, it holds for \textit{measure-and-prepare schemes} or \textit{generalised L\"uders instruments with rank-1 effects}.
Observe that under this identification, the fact that the instrument is balanced at $0$ implies together with the geometrical closure condition that $0 = \sum_{j=1}^n c_j \psi_j = \sum_{j=1}^n \mathcal{I}_j(0)$, meaning that the overall, non-selective measurement process leaves the distinguished state invariant.
Thus, the geometry of $\Psi$ ensures the stability of this state under such measurement update rules, see \cite[Sec. 6.1]{szymusiak2025can} for details.
Indeed, it is in this case that we were able to prove a general form of the so-called \textit{Urgleichung} (\textit{primal equation}), well known from QBism, in \mbox{\cite[Sec. 6.2]{szymusiak2025can}} and \cite[Sec. 5]{szymusiak2025morphophoricity}, see below.

Within this landscape, \textit{informationally complete} (\textit{IC}) measurements, defined by $\lin \Psi = V_0$, are of paramount importance, as they allow for the full operational reconstruction of any state from measurement statistics, i.e., \textit{quantum tomography}. For such measurements, the closure condition implies that $0 \in \interior(\conv\Psi)$.
However, the significance of IC measurements extends far beyond the practical demands of state reconstruction; in certain interpretational frameworks, they constitute the foundational axioms of the theory itself. This geometric perspective proves particularly elegant and natural in the context of QBism~\cite{fuchs2013quantum}. In this approach, ``Urgleichung'' serves as the primary law of probability, relating an agent's beliefs to the outcomes of a \textit{symmetric informationally complete} (\textit{SIC}) measurement. 
Building upon this foundational perspective, recent work has systematically expanded the geometric framework of quantum measurements. This progression began with \textit{morphophoric} quantum measurements~\cite{slomczynski2020morphophoric} -- encompassing \textit{2-design} (rank-1 and equal trace) POVMs such as SICs, MUBs (\textit{Mutually Unbiased Bases}) and MUB-like POVMs, as well as SI-POVMs (equal trace informationally complete POVMs of arbitrary rank~\cite{appleby2007symmetric}). Then, it was naturally extended to \textit{morphophoric} measurements in arbitrary GGPTs~\cite{szymusiak2025can}. Ultimately, this scope was broadened to capture the entire class of \textit{$s$-tight} IC measurements~\cite{szymusiak2025morphophoricity}, a comprehensive family that includes both morphophoric and \textit{tight IC}~\cite{scott2006tight} measurements as special cases, thus completing the line of development opened by the Urgleichung.
One further conceptual point deserves to be stated at the outset: tightness in all its variants is not an absolute attribute of a measurement alone, but of the pair consisting of a measurement and an inner product. The classification developed in this paper is therefore relative to a fixed inner product. In practice this relativity is harmless whenever the theory comes with a distinguished inner product, as quantum mechanics does: there the
Hilbert--Schmidt structure is the natural choice, being compatible with the self-duality of the cone and invariant under unitary
symmetries~~\cite{bengtsson2017geometry}. The broader implications of this relativity, including the construction of Euclidean structures from reference measurements, are deferred to Section~\ref{sec:conclusions}.

In all these scenarios, the crucial mathematical feature underpinning the operational structure is the \textit{tightness} of specific frames associated with the measurement vectors.
Translated into the geometric language introduced above, a measurement represented by a sequence of pairs $(\psi_j, c_j)_{j=1}^n$ is classified as follows: it is morphophoric (resp. tight IC) if $(c_j \psi_j)_{j=1}^n$ (resp. $(\sqrt{c_j} \psi_j)_{j=1}^n$) forms a tight frame for $V_0$, and $s$-tight IC if $(\psi_j)_{j=1}^n$ forms a scalable frame for $V_0$.

At first glance, the operational quest for optimal information extraction
-- rooted in the statistical and informational foundations of physical theories -- seems far removed from the concrete intuition of classical
Euclidean geometry. Yet it is exactly here that a crucial connection emerges. Namely, the primary objective of this paper is to establish a correspondence between a class of \textit{minimal} (i.e., consisting of exactly $d+1$ outcomes, where $d = \dim V_0$) IC measurements, specifically, the \textit{minimal $s$-tight IC measurements}, and the geometry of \textit{orthocentric simplices}. Note that the class of minimal IC measurements, known simply as \textit{MIC}s, plays an important role in quantum mechanics~\cite{debrota2020informationally,debrota2020symmetric,debrota2021varieties}.
Let us observe that for such a measurement $c$ is fully determined as the vector of \textit{barycentric coordinates} of $0$ with respect to $\conv \Psi$ and, moreover, it has a clear geometric interpretation (see Figure~\ref{fig:measurement_probability_unified}).
We reveal that the frame-theoretic condition shaping the state reconstruction formula, the scalability that guarantees a dual frame of the canonical form, dictates a spatial harmony: it physically manifests
as the concurrence of the altitudes of the simplex generated by the \emph{suitably rescaled} measurement vectors, with the orthocentre at the origin.

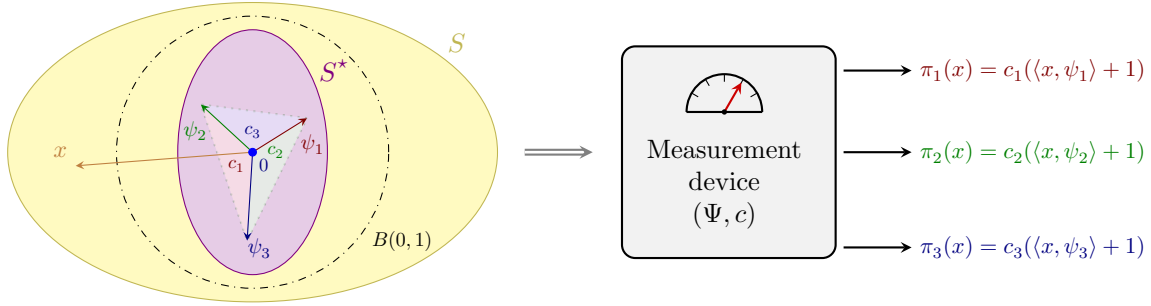
\begin{figure}[htbp]
    \centering
    \begin{tikzpicture}[
        scale=1.8,
        >=stealth,
        box/.style={draw, thick, minimum width=2.2cm, minimum height=2.8cm, rounded corners=5pt, fill=gray!10, align=center}
        ]

        \def\sa{1.8} \def\sb{1.1}
        \def\pax{0.55} \def\pby{0.9}

        \draw[fill=yellow!30, draw=yellow!70!black, thin] (0,0) ellipse ({\sa} and {\sb});
        \node[text=yellow!70!black, font=\bfseries] at (1.5, 0.8) {$S$};

        \draw[black, thin, dash dot] (0,0) circle (1);
        \node[text=black, scale=0.75, font=\bfseries] at (1.1, -0.65) {$B(0,1)$};

        \draw[fill=violet!20, draw=violet, thin] (0,0) ellipse ({\pax} and {\pby});
        \node[text=violet, font=\bfseries] at (0.6, 0.6) {$S^\star$};

        \coordinate (O) at (0,0); 

        \coordinate (P1) at (0.4, 0.25);
        \coordinate (P2) at (-0.38, 0.35);
        \coordinate (P3) at (-0.04, -0.65);

        \draw[gray!60, dotted, thick] (P1) -- (P2) -- (P3) -- cycle;

        \fill[red!10, opacity=0.6] (O) -- (P2) -- (P3) -- cycle;
        \node[text=red!50!black, scale=0.7] at (-0.12, -0.1) {$c_1$};

        \fill[green!10, opacity=0.6] (O) -- (P1) -- (P3) -- cycle;
        \node[text=green!50!black, scale=0.7] at (0.17, 0.0) {$c_2$};

        \fill[blue!10, opacity=0.6] (O) -- (P1) -- (P2) -- cycle;
        \node[text=blue!50!black, scale=0.7] at (0.0, 0.15) {$c_3$};

        \draw[->, thin, red!50!black] (O) -- (P1);
        \node[text= red!50!black, scale=0.8, anchor=west] at (0.3, 0.06) {$\psi_1$};
        
        \draw[->, thin, green!50!black] (O) -- (P2);
        \node[text= green!50!black, scale=0.8, anchor=east] at (-0.28, 0.16) {$\psi_2$};
        
        \draw[->, thin, blue!50!black] (O) -- (P3);
        \node[text=blue!50!black, scale=0.8, anchor=north] at (0.05, -0.58) {$\psi_3$};

        \coordinate (X) at (-1.3, -0.1);
        \draw[->, thin, brown] (O) -- (X) node[above left, scale=0.9] {$x$};
        
        \draw[->, thick, double, double distance=1pt, gray] (2.0, 0) -- (2.5, 0);

     \node[box] (device) at (3.5,0) {
    \begin{tabular}{c} 
        \tikz[scale=0.5]{
            \draw[thick] (-1,0) -- (1,0);
            \draw[thick] (-1,0) arc (180:0:1);
            \foreach \a in {180, 150, 120, 90, 60, 30, 0} \draw[thin] (\a:0.8) -- (\a:1);
            \draw[->, thick, red!80!black] (0,0) -- (60:0.9);
            \fill (0,0) circle (2pt);
        } \\[1ex]
        Measurement \\ 
        device \\ 
        $(\Psi, c)$ 
    \end{tabular}
};

        \draw[->, thick] (device.east) ++(0.05, 0.6) -- ++(0.5,0) node[right, align=left, scale=0.85,red!50!black] {$\pi_1(x) = c_1(\langle x, \psi_1 \rangle + 1)$};
        
        \draw[->, thick] (device.east) ++(0.05, 0.0) -- ++(0.5,0) node[right, scale=0.85,green!50!black] {$\pi_2(x) = c_2(\langle x, \psi_2 \rangle + 1)$};
        
        \draw[->, thick] (device.east) ++(0.05, -0.7) -- ++(0.5,0) node[right, scale=0.85, blue!50!black] {$\pi_3(x) = c_3(\langle x, \psi_3 \rangle + 1)$};
        
\filldraw[blue] (O) circle (0.8pt) node[below right, scale=0.7, text=blue!50!black] {$0$};

    \end{tikzpicture}
    \caption{
    Conceptual geometric representation of states and measurements in $V_0$. The compact convex set of states $S$ contains $0$ in its interior. The measurement vectors $\psi_j$ belong to the polar dual $S^\star$ and satisfy the closure condition, meaning the origin lies strictly within their convex hull. For a minimal measurement $c$ is uniquely determined by $\Psi$ and serves as the set of barycentric coordinates of the origin with respect to $\Psi$. Crucially, in this case barycentric coordinates admit a direct spatial interpretation: each $c_j$ is the ratio of the hypervolume of the sub-simplex formed by the origin and the facet opposite to $\psi_j$ to the total volume of the measurement simplex.}
    \label{fig:measurement_probability_unified}
\end{figure}

A key technical insight of our work is the \textit{rescaling} procedure necessary to reveal this correspondence. A~minimal $s$-tight IC measurement $(\Psi,c)$ yields an \textit{acute orthocentric simplex with its orthocentre at the origin} only after a specific rescaling. The measurement vectors $\psi_j$ must maintain their directions but be scaled by uniquely determined factors. These factors are governed by the ratio of the measurement probabilities $c_j$ to the target barycentric coordinates $p_j$ of the origin. Only under this transformation do they form the vertices of an acute orthocentric simplex $\conv (\frac{c}{p}\Psi)$ centred at the origin. Note that both morphophoric measurements and tight IC measurements lead to orthocentric simplices, 
but the scaling factors differ: we have a trivial scaling ($p_j=c_j$, thus $c_j/p_j=1$) for a tight IC measurement, and a non-trivial scaling ($p_j=1/(d+1)$, thus $c_j/p_j=(d+1)\cdot c_j$) for morphophoric measurements. Consequently, orthocentric simplices corresponding to the latter case must be regular, as their orthocentres and centroids coincide~\cite{hajja2006coincidences}. 
Note that a MIC that is both tight IC and morphophoric is necessarily unbiased.

The cornerstone of our results is established in our main theorem (Theorem~\ref{th:equiv} in Section~\ref{sec:mainresults}), which formalises this rescaling procedure and provides a dictionary between the operational framework of GGPTs and classical convex geometry. Specifically, we prove a three-way equivalence showing that a MIC is $s$-tight if and only if its suitably rescaled measurement vectors form the vertices of an acute orthocentric simplex centred at the origin, which in turn is equivalent to the simplex being homothetically self-dual (i.e., identical to its own polar dual up to a negative homothety). Crucially, this theorem is highly constructive: it provides explicit algebraic formulas linking the scalability constants and measurement probabilities of the measurement to the purely geometric barycentric coordinates and the Gram matrix of the associated simplex. 

Theorem~\ref{th:s_tight_criterion} shows that $s$-tightness of a minimal IC measurement admits a purely angular characterisation.
Namely, a MIC measurement $(\Psi,c)$ is $s$-tight if and only if it fulfils two
conditions:
\begin{itemize}
  \item[(I)] \emph{(obtuse angle condition)}
    $\langle\psi_j,\psi_k\rangle<0$ for $j\neq k$, $j,k=1,\dots,d+1$;
  \item[(II)] \emph{(cross-ratio rule)}
    $\langle\psi_i,\psi_j\rangle\langle\psi_k,\psi_l\rangle
     =\langle\psi_i,\psi_l\rangle\langle\psi_k,\psi_j\rangle$
    for pairwise distinct $i,j,k,l=1,\dots,d+1$.
\end{itemize}

This theorem thus recasts measurement minimality and tightness as purely angular conditions on the measurement vectors. 
As we explain in Section~\ref{sec:skeletons}, this characterisation can be pushed to a complete classification: modulo rotations, the admissible direction configurations are parametrised bijectively by the open probability simplex $\Delta^{\circ}_{d+1}$ -- each configuration being encoded by a single probability vector, which we call the \emph{skeleton} of the measurement -- together with an orientation class. A skeleton, an orientation, and a sequence of admissible lengths determine the measurement uniquely.

The paper is organised as follows. In Section~\ref{sec:newapproach}, we introduce the equivalent geometric approach to \mbox{GGPTs} and recall basic facts concerning finite frames. 
Section~\ref{sec:orthocentric} is devoted to the geometry of orthocentric simplices and their subclasses. 
In Section~\ref{sec:mainlemma}, we prove the main lemma connecting scalable frames to homothetically self-dual orthocentric simplices.
Section~\ref{sec:mainresults} establishes the core equivalence theorems between measurements and simplices, in both directions. 
In Section~\ref{sec:algebraiccriterion}, an algebraic criterion for $s$-tightness of MICs is introduced. Section~\ref{sec:examples} provides examples, including an explicit analysis of the qubit case.
Section~\ref{sec:conclusions} concludes; in particular, Section~\ref{sec:skeletons} summarises the resulting classification of minimal $s$-tight IC measurements by their skeletons. 
Appendix~\ref{app:instrument} relates the affinity of instruments to the Bayes rule, and Appendix~\ref{app:equivalence} proves the equivalence of the internal-geometry approach with the standard formulation of GGPTs.

\section{Equivalent approach to GGPT} \label{sec:newapproach}

There are three basic ingredients in any \textit{generalised probabilistic theory} (\textit{GPT}): \textit{states} ($B$), \textit{effects} ($\mathcal E$) and a \textit{function assigning probabilities} $\mathfrak{p}:B\times\mathcal E\to [0,1]$. The standard approach is to introduce them in the language of \textit{ordered vector spaces}, i.e., the whole theory can be derived from a triple $(V,C,B)$, where $V$ is a vector space ordered by the convex cone $C$ and $B$ is a convex base of $C$. Since such a base is necessarily contained in the 1-level set of a unique positive functional $e\in C^*$ and every such functional defines a unique convex base of $C$, one can equivalently start with $(V,C,e)$. For the purpose of this paper we assume that $V$ is finite dimensional and $\dim V=d+1$. Now, the states are represented by the elements of $B$ and the set of effects $\mathcal E$ consists of the elements of the dual cone, $g\in C^*$, such that $g\leq e$. Finally, the probability assignment works simply as $\mathfrak p(x,g)=g(x)$. In this setup we can also define a \textit{measurement} as a sequence of non-zero effects $(g_j)_{j=1}^n$ such that $\sum_{j=1}^n g_j=e$. 

If the vector space $V$ is additionally equipped with an \textit{inner product} $\langle\cdot,\cdot\rangle$, we call such a GPT \textit{geometric} (\textit{GGPT}). In such a case, we can decompose $V$ as an orthogonal direct sum of $V_0=\aff B -\aff B$ and $\lin(\{m\})$, where $m\in \aff B$ is a unique vector orthogonal to $V_0$. The inner product thus consists of two parts: the inner product on $V_0$, $\langle\cdot,\cdot\rangle_0 \coloneqq \langle\cdot,\cdot\rangle_{|_{V_0\times V_0}}$, and the indication of the \textit{orthogonal direction}, $m$, together with the \textit{size parameter} $\mu \coloneqq \lVert m \rVert^2$. The first part is an intrinsic property of the theory, as it characterises the \textit{internal geometry} of the set of states $B$. The second one describes the \textit{external geometry} of $B$, i.e., how it is positioned in the space $V$ (with respect to $\langle\cdot,\cdot\rangle$). Consequently, the starting point for making a GPT geometric can be to fix the inner product on $V_0$ and only then choose $m \in \aff B$ and $\mu>0$ in a way that we find the most appropriate. In particular, we can assume that $m\in\interior B$ and $\mu$ is such that the inner product is compatible with the order structure in one of the following ways: $C\subset C^+$ or $C^+\subset C$, where $C^+ \coloneqq \{y\in V:\langle x,y\rangle\geq 0 \quad \text{for all }x\in C\}$ is the \textit{positive dual cone} of $C$ in $V$. In the former case we say that the cone $C$ (and the theory) is \textit{infra-dual} while in the latter one we call them \textit{supra-dual}. If it is possible to make the theory both infra- and supra-dual, we call it then self-dual (e.g., both classical and quantum theories are self-dual).

The relation between the cone $C$ and its positive dual $C^+$ can also be expressed in terms of dual sets. For a convex subset $A$ of an affine space $\mathcal A=\aff A$ we define its \textit{dual set} $A_{a,r}^\star$ in $\mathcal A$ with respect to the sphere with centre at $a\in\mathcal A$ of radius $r>0$ as $A_{a,r}^\star\coloneqq \{x\in\mathcal A:\langle x-a,y-a\rangle\geq-r^2 \quad \text{for all }y\in A\}$, also called \textit{dual polar} or \textit{negatively polar set} by some authors. In particular, we will use $A^\star \coloneqq A^\star_{0,1}$. Note that if $A$ is a polytope and $a \in \interior A$, then $A_{a,r}^\star$ is also a polytope. It is easy to see \cite[Proposition 1.ii.]{szymusiak2025can} that $B^\star_{m,\sqrt{\mu}}$ is a convex base of the positive dual cone $C^+$. Consequently, infra-, supra- and self-duality can be expressed as appropriate inclusions between $B$ and $B^\star_{m,\sqrt{\mu}}$: $B \subset B^\star_{m,\sqrt{\mu}}$, $B^\star_{m,\sqrt{\mu}} \subset B$, and $B = B^\star_{m,\sqrt{\mu}}$, respectively.

If the GGPT is supra-dual, the isometric linear isomorphism $T:V\to V^*$ induced by the inner product via $T(y)(x)=\langle x,y\rangle$ allows the identification of effects with the (unnormalised) states. Indeed, since \mbox{$T(C^+)=C^*$,} $T^{-1}(g)\in C^+\subset C$ for any $g\in C^*$. In particular, $T^{-1}(e)=\frac{1}{\mu}m$. As a consequence, the probability assignment for state $x$ and effect $g$ can also be expressed by the inner product: $\mathfrak p(x,g)=g(x)=\langle x,T^{-1}(g)\rangle$. In particular, $g(m)=\langle m, T^{-1}(g)\rangle=\mu\langle T^{-1}(e),T^{-1}(g)\rangle=\mu e(T^{-1}(g))$. 

In the following, we propose an approach to GGPTs that focuses on the intrinsic properties of the theory, the internal geometry. 
It turns out that we can build a GGPT from just three components: a \textit{finite-dimensional vector space} $V_0$, an \textit{inner product} $\langle\cdot,\cdot\rangle_0$ on $V_0$, and a \textit{convex closed neighbourhood of} $0$, $S\subset V_0$. This approach can be viewed as a generalisation of the \textit{Bloch representation} in quantum theory, but extending it to effects and measurements as well.

The path from the GGPT $(V,C,B,\langle\cdot,\cdot \rangle_0,m,\mu)$ to this new representation $(V_0,\langle\cdot,\cdot\rangle_0,S)$ is as follows. Obviously, $V_0$ and the inner product in it are already defined. The set $S$ now represents states and is defined as $S \coloneqq \frac{1}{\sqrt{\mu}}(B-m)=\frac{1}{\sqrt{\mu}}P_0(B)$, where $P_0:V\to V_0$ is the orthogonal projection onto $V_0$. It resembles the Bloch representation in the sense that we replace objects from an affine space with the ones from the vector space of the same dimension. However, we are going to take it one step further and represent effects by using the elements of $V_0$. 

To be precise, since the set of effects is no longer contained in any non-trivial affine space,  we need a pair $(\psi,c)\in V_0\times[0,1]$  to represent an effect $g$. Indeed, we use the isomorphism $T^{-1}$ to consider the effect first as an element of $V$ and put $v \coloneqq T^{-1}(g)$, and then decompose it in the following way: $$v=(v-e(v)m)+e(v)m=P_0(v)+\frac{g(m)}{\mu} m=\frac{g(m)}{\sqrt{\mu}}\left(\frac{1}{\sqrt{\mu}}P_0\left(\frac{\mu}{g(m)}v\right)+\frac{1}{\sqrt{\mu}}m\right).$$ Note that $\frac{\mu}{g(m)}v\in \aff B$.  It is now easy to see that by putting $\psi \coloneqq \frac{1}{\sqrt{\mu}}P_0\left(\frac{\mu}{g(m)}v\right)$ and $c \coloneqq g(m)$   we uniquely describe the effect $g$, see Figure~\ref{fig:newggpt}. While the normalisation factors might look counter-intuitive, they result in the elegant $\mu$-independent description of effects:
\begin{proposition} \label{pr:effect}
    For a pair $(\psi,c)\in (V_0\times(0,1])\cup\{(0,0)\}$ the corresponding vector $g=T(\frac{c}{\sqrt{\mu}}\psi+\frac{c}{\mu}m)$ is an effect if and only if $\psi\in S^\star$ and $\frac{c}{c-1}\psi\in S^\star$ if $c\in(0,1)$, or $\psi=0$ if $c=1$; for $c=0$, $g=0$ is trivially an effect.
\end{proposition}
\begin{proof} Let us assume that $c\in(0,1]$ and let $x\in B$. Then 
    \begin{align*}
        g(x)&=\left\langle \frac{c}{\sqrt{\mu}}\psi+\frac{c}{\mu}m,x\right\rangle=\left\langle\frac{c}{\sqrt{\mu}}\psi+\frac{c}{\mu}m,(x-m)+m\right\rangle\\
        &=c\left\langle\frac{1}{\sqrt{\mu}}\psi,x-m\right\rangle+c=c\left(\left\langle\psi,\frac{1}{\sqrt{\mu}}(x-m)\right\rangle+1\right).
    \end{align*}
    Since $x\in B$ if and only if $\frac{1}{\sqrt{\mu}}(x-m)\in S$, the condition $g(x)\in[0,1]$ for all $x\in B$ can be now stated as $-1\leq\langle\psi,y\rangle\leq\frac{1-c}{c}$ for all $y\in S$. The left-hand side inequality means that $\psi\in S^\star$. For $c=1$ the right-hand side inequality means that $\psi=0$ while for $c<1$ it can be expressed as $\frac{c}{c-1}\psi\in S^\star$.
 \end{proof}

We can now describe the set of effects as $\mathcal E \coloneqq \{(\psi,c)\in S^\star\times(0,1):\frac{c}{c-1}\psi\in S^\star\}\cup\{(0,0),(0,1)\}$.
The function assigning probability now takes the form
\begin{equation} \label{eq:prob-affine}
    \mathfrak{p}(y,(\psi,c))=c(\langle\psi,y\rangle+1)
\end{equation}
 for $y\in S$ and $(\psi,c)\in\mathcal E$. As a corollary we also obtain that the theory is supra- (infra-) dual whenever $S^\star\subset S$ ($S\subset S^\star$, respectively).

\begin{figure}[htb]
    \centering
    \includegraphics[width=0.8\linewidth]{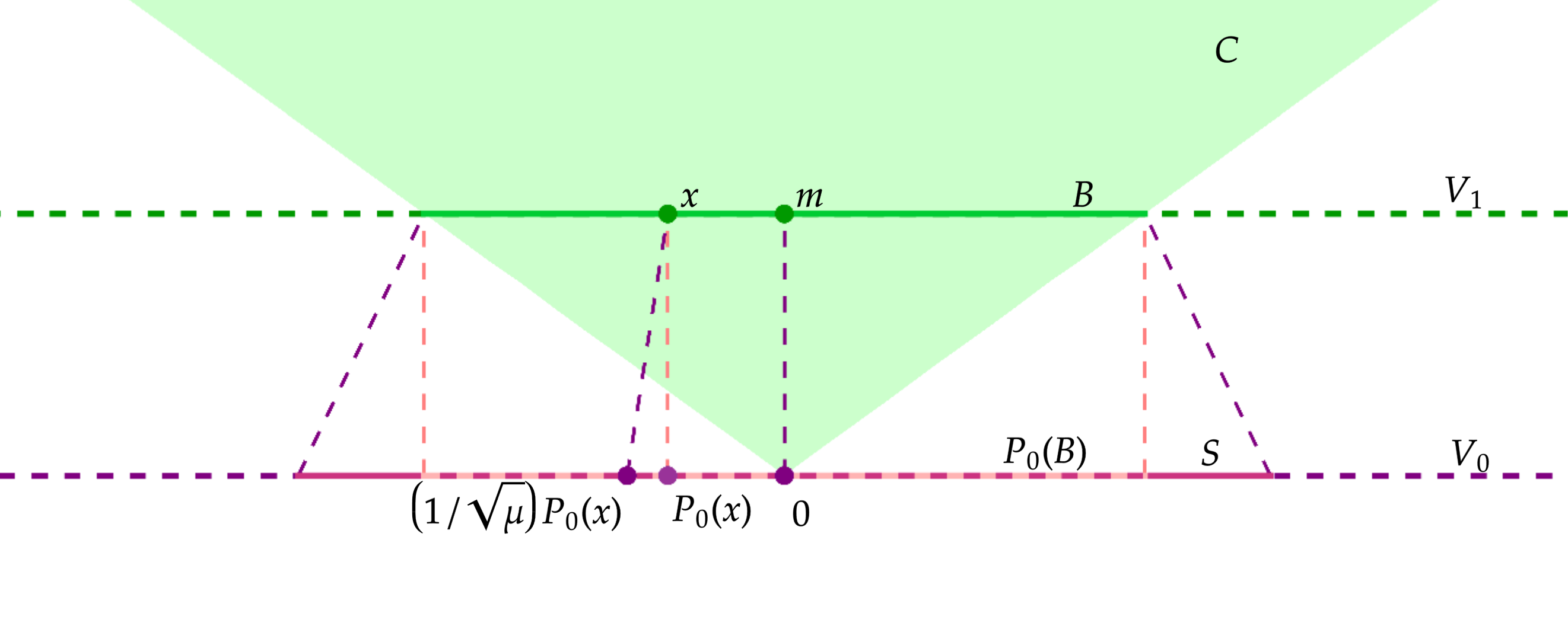}
    \includegraphics[width=0.8\linewidth]{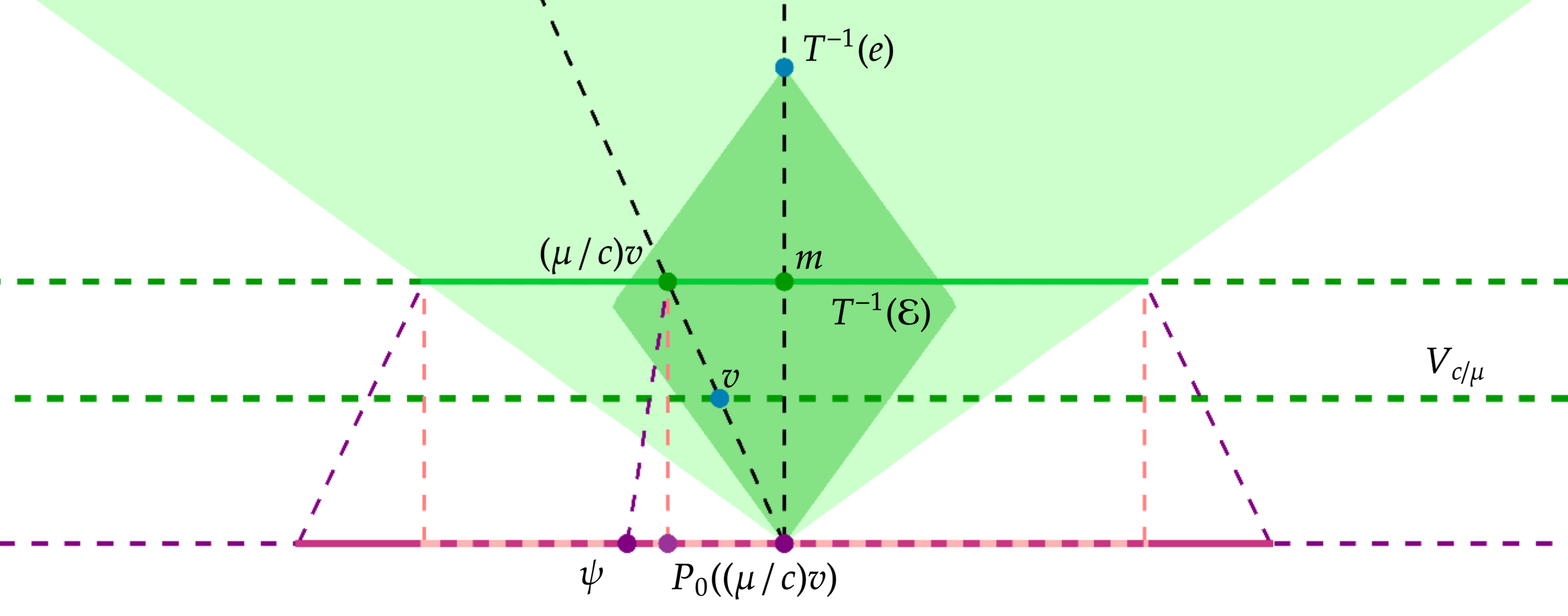}  
    \caption{Above: The standard description of states is by a convex base $B$ of a cone $C$. In our approach states are projected from $V_1=\aff B$ to the corresponding linear subspace $V_0$ and then rescaled by $1/\sqrt{\mu}$ in order to make the theory independent of the norm of the distinguished state $m$. Below: The effects can be transferred from the dual space to the space $V$ via the unique isometric isomorphism $T^{-1}$. An effect $g$ mapped by $T^{-1}$ to $v$ can be uniquely represented by the ray containing it and the affine subspace $V_{c/\mu}$ parallel to $V_0$, where $c=g(m)$. Instead of the whole ray we can consider its intersection with $V_1$, i.e., the point $(\mu/c)v$, and, in consequence, the whole information about the effect $g$ can be stored as a pair $(\psi,c)$, where $\psi=(1/\sqrt{\mu})P_0((\mu/c)v)$ is constructed in the same way as the states in our approach.}
    \label{fig:newggpt}
\end{figure}

Finally, two simple conditions are required  in order to describe a measurement:
\begin{proposition}
    A sequence $(\psi_j,c_j)_{j=1}^n\subset V_0\times (0,1]$ forms a measurement if and only if $\psi_j\in S^\star$ for $j=1,\dots,n$, $\sum_{j=1}^nc_j=1$, and $\sum_{j=1}^nc_j\psi_j=0$.
\end{proposition}
\begin{proof}
    We have already shown that $T(\frac{c_j}{\sqrt{\mu}}\psi_j+\frac{c_j}{\mu}m)\in C^*$ if and only if $\psi_j\in S^\star$. To complete the proof it suffices to observe that
    \begin{align*}
        \sum_{j=1}^nT(\frac{c_j}{\sqrt{\mu}}\psi_j+\frac{c_j}{\mu}m)=e &\iff \frac{1}{\sqrt{\mu}}\sum_{j=1}^nc_j\psi_j+\frac{1}{\mu}\sum_{j=1}^nc_jm=\frac{1}{\mu}m\\ &\iff \sum_{j=1}^nc_j\psi_j=0 \textnormal{ and }\sum_{j=1}^nc_j=1.\qedhere
    \end{align*}
\end{proof}

\subsection{Measurements (frames, classes)}

In order to introduce the class of $s$-tight IC measurements and its subclasses of tight IC and morphophoric measurements, let us start with recalling some basic definitions and facts concerning finite frames (for general background see, e.g.,~\cite{waldron2018introduction}, scalable frames:~\cite{kutyniok2013scalable}, probabilistic tight frames:~\cite{ehler2013probabilistic}).

\begin{definition} Let  $\mathcal H$ be a $d$-dimensional real inner product space.
    We say that a sequence $(h_j)_{j=1}^n$ of vectors in $\mathcal H$ is
    \begin{itemize}
        \item a \textit{frame} for $\mathcal H$ if there exist $A,B>0$ such that 
        \begin{equation} \label{eq:frame}
            A \lVert h \rVert^2\leq\sum_{j=1}^n|\langle h,h_j\rangle|^2\leq B \lVert h \rVert^2\quad \textnormal{for all }h\in\mathcal H,
        \end{equation}
        \item a \textit{tight frame} for $\mathcal H$ if there exists $A>0$ (called frame bound) such that 
        \begin{equation} \label{eq:tightframe}
            \sum_{j=1}^n|\langle h,h_j\rangle|^2= A\lVert h \rVert^2\quad \textnormal{for all }h\in\mathcal H,
        \end{equation}
        or, equivalently 
        \begin{equation} \label{eq:tightframe2}
           \frac{1}{A} \sum_{j=1}^n\langle h,h_j\rangle h_j= h\quad \textnormal{for all }h\in\mathcal H,
        \end{equation}
        \item a \textit{(strictly) scalable frame} for $\mathcal H$ if there exist positive constants $(s_j)_{j=1}^n$ such that $(s_jh_j)_{j=1}^n$ is a tight frame for $\mathcal H$.
    \end{itemize}
    Two frames $(h_j)_{j=1}^n$ and $(h_j')_{j=1}^n$ are said to be \textit{dual} if
    \begin{equation} \label{eq:dualframe}
        h=\sum_{j=1}^n\langle h, h_j'\rangle h_j=\sum_{j=1}^n\langle h, h_j\rangle h_j' \quad \textnormal{for all } h\in\mathcal H.
    \end{equation}
\end{definition}

The intuitions behind these notions are as follows: inequalities \eqref{eq:frame} mean that $(h_j)_{j=1}^n$ is nothing more than a spanning sequence for $\mathcal H$. Eqs. \eqref{eq:tightframe} and \eqref{eq:tightframe2} can be viewed respectively as a version of Parseval's identity and a Fourier-type expansion,  which in turn allows tight frames to be interpreted as natural generalisations of orthonormal bases. 
Finally, the scalability of a frame guarantees the existence of a dual frame of a specific form, namely if the frame $(h_j)_{j=1}^n$ is scalable with the scalability constants $(s_j)_{j=1}^n$ to a tight frame with the frame bound $A$, then $(s_j^2A^{-1}h_j)_{j=1}^n$ is its dual frame. 

Scalable frames are strictly related to probabilistic tight frames:

\begin{definition}
    A Borel probability measure $\nu$ on $\mathcal H$ is called a \textit{probabilistic tight frame} if there exists $A>0$ such that \begin{equation*}
        \int_{\mathcal H}|\langle h, g\rangle|^2d\nu(g) = A \lVert h \rVert^2.
    \end{equation*}
\end{definition}
   
Let $\Delta_n$ denote the \emph{probability simplex}, i.e., $\Delta_n \coloneqq \{p\in\mathbb R^n:p_j\geq0 \textnormal{ for }j=1,\dots,n \textnormal{ and }\sum_{j=1}^n p_j=1\}$ and $\Delta_n^\circ$ its relative interior, i.e., $\Delta_n^\circ \coloneqq \{p\in \Delta_n:p_j>0 \textnormal{ for }j=1,\dots,n \}$. If $\nu$ is a probability measure on $\mathcal H$ with finite support, i.e., $\nu \coloneqq \sum_{j=1}^np_j\delta_{h_j}$ for some $p=(p_1,\dots,p_n)\in\Delta_n^\circ$ and $(h_j)_{j=1}^n\subset\mathcal H$, then $\nu$ is a probabilistic tight frame if and only if $(\sqrt{p_j}h_j)_{j=1}^n$ is a tight frame with frame bound $A$. As a consequence, if $(h_j)_{j=1}^n$ is a scalable frame with scalability constants $(s_j)_{j=1}^n$, then $\left(\sum_{k=1}^ns_k^2\right)^{-1}\sum_{j=1}^n s_j^2\delta_{h_j}$ is a probabilistic tight frame.

\begin{remark} \label{re:nonuniqueconstants}
    The scalability constants $(s_j)_{j=1}^n$ in the definition of a scalable frame are not unique. Obviously, we can always multiply them by a common factor. Furthermore, whenever there exist some scalable subframes, any convex combination with full support of the scalability constants for such subframes is a scalability constants sequence for the frame in question. In fact, the set of scalability constants is a convex cone with a polyhedral convex base. The vertices of this polyhedron are given by the scalability constants for subframes that are scalable but do not contain any scalable subframe themselves.
\end{remark}

One of the most useful tools to determine whether a sequence of vectors forms a tight frame  is the \textit{variational characterisation} \cite[Theorem 6.1]{waldron2018introduction}:

\begin{theorem}[variational characterisation] \label{th:variationalcharacterisation}
    Let $H \coloneqq (h_j)_{j=1}^n$ be a sequence of vectors in $\mathcal H$ (not all zero). Then 
    \begin{equation*}
      v(H) \coloneqq  \sum_{j=1}^n\sum_{k=1}^n|\langle h_j,h_k\rangle|^2- \frac{1}{d}\left(\sum_{j=1}^n \lVert h_j \rVert^2\right)^2\geq 0
    \end{equation*}
    and equality holds if and only if $(h_j)_{j=1}^n$ is a tight frame for $\mathcal H$.
\end{theorem}

The frame language turns out to be very useful in describing some special classes of measurements. 
Measurements with the property that the statistics of their outcomes uniquely determine the pre-measurement state are of particular importance in any GPT and are known as \textit{informationally complete} (\textit{IC} for short) measurements. In the standard approach, such a measurement is characterised by the requirement that its effects span the dual space $V^*$. This property can be translated into our approach by demanding that the measurement vectors $\psi_j$ span $V_0$. 

The notions of $s$-tight IC measurements and of the subclasses of tight IC and morphophoric measurements have been stated originally in the language of the standard approach in~\cite{szymusiak2025can,szymusiak2025morphophoricity}. The requirement is always some kind of scalability for the frame given by $(P_0(v_j))_{j=1}^n$, where the effects $g_j=T(v_j)$ form the measurement in question. Since in the new approach $\psi_j$ are proportional to $P_0(v_j)$, the scalability requirement remains but with updated scalability constants:

\begin{definition}
    The measurement  $(\psi_j,c_j)_{j=1}^n$ is said to be
    \begin{itemize}
    \item \textit{informationally complete (IC)} if $(\psi_j)_{j=1}^n$ is a  frame for $V_0$,
    \item \textit{minimal IC (MIC)} if it is IC and $n=d+1$,
            \item \textit{$s$-tight IC} if $(\psi_j)_{j=1}^n$ is a scalable frame for $V_0$,
        \item \textit{tight IC} if $(\sqrt{c_j}\psi_j)_{j=1}^n$ is a tight frame for $V_0$,
        \item \textit{morphophoric} if $(c_j\psi_j)_{j=1}^n$ is a tight frame for $V_0$.
    \end{itemize}
\end{definition}

\begin{remark} \label{re:scalability}
If the measurement $(\psi_j,c_j)_{j=1}^n$ is $s$-tight IC, the
scalability constants are in general not unique, as explained in
Remark~\ref{re:nonuniqueconstants}.
However, they become unique (up to rescaling by a common factor) in the
minimal case, provided that the dimension $d\geq 2$.
Indeed, according to \cite[Theorem 3.1]{chan2017minimal} it is
sufficient to show that $\Psi$ has no proper scalable subframe.
Any subframe must span $V_0$, hence contain at least $d$ vectors, so a
proper one consists of exactly $d$ of them; relabelling the outcomes, we
may assume it is $(\psi_j)_{j=1}^{d}$.
Being a tight frame of $d$ vectors after rescaling, it must be an
orthogonal basis, i.e., $(\psi_j)_{j=1}^{d}$ are pairwise orthogonal.
However, it follows from \cite[Theorem 3.1]{chan2017minimal} and
$s_{d+1} > 0$ that there exists a proper scalable subframe of $\Psi$
containing $\psi_{d+1}$; by the above it consists of $d$ vectors, so it
omits exactly one $\psi_i$ with $i\leq d$ -- say $i=1$, after
relabelling the remaining indices.
Then $(\psi_j)_{j=2}^{d+1}$ are pairwise orthogonal as well. Hence both
$\psi_1$ and $\psi_{d+1}$ are orthogonal to
$\lin(\psi_2,\dots,\psi_d)$, whose orthogonal complement in $V_0$ is
one-dimensional, so $\psi_{d+1}=a\psi_1$ for some $a\neq 0$.
From the closure condition we now obtain that $\psi_1,\dots,\psi_d$ are
linearly dependent, a contradiction.
Therefore, $s$ is uniquely determined by $\Psi$ up to rescaling by a
common factor.
\end{remark}

\section{Orthocentric simplices} \label{sec:orthocentric} 

That the altitudes of any planar triangle intersect at a single point is a classical result known since ancient Greek geometry. However, this elegant property does not trivially carry over to higher dimensions. In fact, as observed by Steiner~\cite{steiner1826fortsetzung} as early as 1826, the altitudes of a general $d$-simplex ($d \ge 3$) are not necessarily concurrent. 
The geometry of orthocentric simplices -- the special class of simplices that do preserve this concurrency -- has a rich mathematical history. According to Court~\cite{court1934notes}, the 3-dimensional case, originally termed a ``rectangular'' or ``orthogonal'' tetrahedron, was first considered by Lhuilier in 1782~\cite{lhuilier1782relatione}. The contemporary term \textit{orthocentric tetrahedron} was later proposed by de Longchamps in 1890~\cite{longchamps1890tetraedre}. The formal generalisation to arbitrary $n$-dimensional simplices was pioneered by Schoute at the turn of the 20th century~\cite{schoute1902mehrdimensionale}, and systematically codified by Sommerville~\cite{sommerville1929introduction}.
In the 20th century, these structures were further explored by Egerv\'{a}ry~\cite{egervary1940orthocentric}, Fiedler~\cite{fiedler1956geometrie}, and Gerber~\cite{gerber1975orthocentric}, who revealed their unique geometric properties. For a comprehensive modern treatment, and for the perspective that orthocentric simplices serve as the ``true'' multi-dimensional generalisations of triangles, we refer the reader to Edmonds et al.~\cite{edmonds2005orthocentric}, Fiedler~\cite{fiedler2011matrices}, and Hajja and Martini~\cite{hajja2013orthocentric}.
In this paper we show that a concept originating in 18th- and 19th-century classical geometry, studied for generations primarily for its intrinsic elegance, finds an unexpected application in the modern foundations of quantum theory.

From now on, we assume that $\mathcal{H}$ is a real inner product space of dimension $d \geq 2$.
A \textit{(non-degenerate) $d$-simplex} $\Sigma$ in $\mathcal{H}$ is the convex hull of $d+1$ affinely independent points $R_{1},\dots,R_{d+1} \in \mathcal{H}$, i.e., $\Sigma = \conv (\{R_{1},\dots,R_{d+1}\})$. 
The points $R_{1},\dots,R_{d+1}$ are called the \textit{vertices} of the simplex $\Sigma$ and the line segments connecting distinct vertices are called \textit{edges}. 
The \textit{centroid} of $\Sigma$ is the average of its vertices, $(R_1 + \dots + R_{d+1})/(d+1)$. 
A simplex is called \textit{regular} if and only if all its edges have the same length; then, any permutation of its vertices gives an isometry of the simplex \cite[Proposition 9.7.1]{berger1994geometry}. 
The \textit{facets} of $\Sigma$ are \mbox{$(d-1)$-simplices} whose vertices are arbitrary $d$ vertices of $\Sigma$. 
The \textit{altitudes} of the simplex $\Sigma$ are line segments from a vertex $R$ that are perpendicular to the affine span of the facet $F$ determined by the remaining vertices of $\Sigma$, i.e., the facet opposite to $R$. For a given vertex $R$ and the facet $F$ opposite to it, the altitude from $R$ to $F$ is the shortest line segment from the vertex $R$ to the affine span of $F$. The point of intersection between this altitude and $\aff F$ is the orthogonal projection of the vertex $R$ onto the facet $F$.

In the case of triangles, which are $2$-simplices, the point in which the altitudes intersect is called the orthocentre of the triangle. 
The fact that for $d \ge 3$ the altitudes of a $d$-simplex need not be concurrent distinguishes a special class of simplices, which can be considered as the generalisations of triangles to higher dimensions.
The simplex is called \textit{orthocentric} if its altitudes have a common point, which is then called the \textit{orthocentre} of the simplex.
An orthocentric simplex is called \textit{rectangular} if its orthocentre coincides with one of the vertices and \textit{non-rectangular} (or \textit{oblique}~\cite{kabluchko2026angles}) otherwise. A triangle is rectangular if and only if it is right.

Let $\Sigma = \conv (\{R_{1},\dots,R_{d+1}\}) \subset \mathcal{H}$ be a $d$-simplex in a $d$-dimensional real inner product space $\mathcal{H}$, i.e., the vertices $R_{1},\dots,R_{d+1} \in \mathcal{H}$ are affinely independent. 
The following theorem characterises orthocentric simplices.
\begin{theorem}[{{\cite[Theorem 3.1]{edmonds2005orthocentric}}}] \label{th:orthocentric}
    Let $\Sigma = \conv (\{R_{1},\dots,R_{d+1}\}) \subset \mathcal{H}$ be a $d$-simplex, where $d \coloneqq \dim \mathcal{H} \ge 2$.
    \begin{enumerate}[(a)]
    \item $\Sigma$ is orthocentric if and only if for every $k = 1,\dots,d+1$ the quantity $\langle R_{i}-R_{k}, R_{j}-R_{k} \rangle$ does not depend on $i,j = 1,\dots,d+1$ as long as $i,j,k$ are pairwise distinct. 
    \item If $G \in \aff \Sigma$, then $\Sigma$ is orthocentric with orthocentre $G$ if and only if $\langle R_{j} - G, R_{k} - G \rangle$ does not depend on $j,k = 1,\dots,d+1$ as long as $j \ne k$. In such a case, this quantity equals zero if and only if $\Sigma$ is rectangular.
    \end{enumerate}
\end{theorem}
In particular, if $\Sigma = \conv (\{R_{1},\dots,R_{d+1}\})$ is an orthocentric $d$-simplex with orthocentre $G$, the quantity 
\begin{equation} \label{eq:sigma}
    \sigma \coloneqq -\langle R_{j} - G, R_{k} - G \rangle, \; j \ne k
\end{equation}
is well-defined and $\sigma = 0$ if and only if $\Sigma$ is rectangular. 
If $\Sigma$ is non-rectangular, the barycentric coordinates of the orthocentre $G$ with respect to $\Sigma$, i.e., the unique numbers $p_{1},\dots,p_{d+1} \in \mathbb{R}$ (not necessarily positive) such that
\begin{equation*}
    \sum\limits_{j=1}^{d+1}p_{j}R_{j} = G, \quad \sum\limits_{j=1}^{d+1}p_{j} = 1, 
\end{equation*}
satisfy $p_{j} \notin \{0,1\}$ and we have \cite[Theorem 3.3]{edmonds2005orthocentric}: 
\begin{align} 
    \lVert R_{j} - G \rVert^{2} &= \sigma \frac{1-p_{j}}{p_{j}} \label{eq:vertexnorm} \\
    \lVert R_{j} - R_{i} \rVert^{2} &= \sigma \left( \frac{1}{p_{j}} + \frac{1}{p_{i}} \right) \label{eq:edgenorm}
\end{align}
where $i,j = 1,\dots,d+1$ and $j \neq i$.
Thus, the barycentric coordinates $p_{1},\dots,p_{d+1}$ and the quantity $\sigma \neq 0$ determine the distances $\lVert R_{j} - G \rVert$ and the angles between the vectors $R_{j}-G$ for $j = 1,\dots,d+1$, for a non-rectangular orthocentric simplex with the orthocentre $G$. 

Let us note that in the two-dimensional case, the criterion in Theorem~\ref{th:orthocentric}(a) is always satisfied, as there are only three possible values for each of the indices $i,j,k$.
Furthermore, it follows from Theorem~\ref{th:orthocentric}(a) that every regular simplex is orthocentric.
Indeed, let us assume that $\Sigma = \conv (\{R_{1},\dots,R_{d+1}\})$ is a regular simplex. 
For any triple of pairwise distinct indices $i,j,k \in \{1,\dots,d+1\}$, it follows from the regularity of the simplex $\Sigma$ that there exists an isometry $\iota$ such that $\iota(R_1) = R_i$, $\iota(R_2) = R_j$, $\iota(R_3) = R_k$. 
Therefore, $\langle R_i - R_k, R_j - R_k \rangle = \langle \iota(R_1 - R_3), \iota(R_2 - R_3) \rangle = \langle R_1 - R_3, R_2 - R_3 \rangle$ is constant as long as $i,j,k$ are pairwise distinct.

By eq.~\eqref{eq:edgenorm}, a non-rectangular orthocentric simplex $\Sigma$ is regular if and only if the barycentric coordinates are equal, $p_{j} = 1/(d+1)$. Moreover, by eq.~\eqref{eq:vertexnorm}, this condition is equivalent to $\lVert R_{j}-G \rVert=\text{const}$. Rectangular orthocentric simplices cannot be regular, as the edges that contain the orthocentre are shorter than those which do not.

The signs of the barycentric coordinates of a non-rectangular orthocentric simplex are related to the sign of $\sigma$ and to the vertex angles of $\Sigma$.
\begin{definition}
    Let $R_{k}$ be one of the vertices of a simplex $\Sigma = \conv (\{R_{1},\dots,R_{d+1}\})$. The \textit{vertex angle} at $R_{k}$ is the polyhedral angle with the vertex at $R_{k}$ and the arms being the edges $R_{k}R_{j}$ of the simplex $\Sigma$, where $j \neq k$. This vertex angle is called \textit{strongly acute} (respectively, \textit{right}, \textit{strongly obtuse}) if and only if all angles $\angle(R_{i}R_{k}R_{j})$ are acute (respectively: right, obtuse) for $i,j \neq k$, $i \neq j$. 
\end{definition}
In general, these three classes of vertex angles do not exhaust all possibilities  --  consider a tetrahedron (3-simplex) built from a right triangle by adding a fourth vertex lying above its relative interior. Then, there are two acute angles and one right angle at one of the vertices of the initial triangle (the one at which the angle is right).
However, the vertex angles of a non-rectangular orthocentric simplex are either strongly acute or strongly obtuse. 
\begin{proposition}[{{\cite[Theorem 3.8]{edmonds2005orthocentric}}}] \label{pr:barsigns}
    The numbers $p_{1},\dots,p_{d+1} \in \mathbb{R} \setminus \{0\}$ with $\sum_{j=1}^{d+1}p_{j}=1$ occur as the barycentric coordinates of the orthocentre of a non-rectangular orthocentric simplex if and only if they are all positive, or if exactly one of these numbers is positive and the others negative. The first case corresponds to $\sigma > 0$ and all the vertex angles of $\Sigma$ being strongly acute, while the second case is equivalent to $\sigma < 0$ and one of the vertex angles being strongly obtuse and the others being strongly acute.
\end{proposition}
The fact that all of the barycentric coordinates are positive is equivalent to $G \in \interior \Sigma$, which means that the orthocentre lies in the interior of the simplex. In the second case, the orthocentre lies outside the simplex. In particular, if an orthocentric simplex $\Sigma$ is non-rectangular, the orthocentre cannot lie in the boundary of $\Sigma$. 
Due to the relation between the sign of $\sigma$ and the vertex angles of $\Sigma$, the quantity $-\sigma$ was called the \textit{obtuseness} of $\Sigma$ in~\cite{edmonds2005orthocentric}. However, as the authors of~\cite{edmonds2005orthocentric} noticed, the obtuseness defined in this way is scale-dependent, as multiplying each vector $R_{j}$ by some positive factor $\lambda$ changes $\sigma$ by a factor of $\lambda^{2}$, though its sign is clearly invariant under scaling.  
Proposition~\ref{pr:barsigns} suggests the following definitions.
\begin{definition}
    A non-rectangular orthocentric simplex $\Sigma$ is called 
    \begin{itemize}
        \item \textit{acute}, if all of its vertex angles are strongly acute;
        \item \textit{obtuse}, if one of its vertex angles is strongly obtuse and the others are strongly acute.
    \end{itemize}
\end{definition}
Accordingly, orthocentric simplices can be divided into three pairwise disjoint classes, some of whose properties are summarised in Table~\ref{tab:orttable}. 
In particular, every regular simplex has to be acute orthocentric due to the property $p_{j}=\text{const}$. 
In what follows, we will mainly consider acute orthocentric simplices. 
\begin{table}[H]
\centering
\begin{tabular}{lccc}
    \toprule
    Class & Acute & Rectangular & Obtuse \\
    \midrule
    Sign of $\sigma$ & $\sigma > 0$ & $\sigma = 0$ & $\sigma < 0$ \\
    \addlinespace
    Orthocentre $G$ & interior point & vertex & exterior point \\
    \addlinespace
    Barycentric & \multirow{2}{*}{all positive} & \multirow{2}{*}{all but one equal $0$} & one negative, \\
    coordinates $p$ & & & the others positive \\
    \addlinespace
    \multirow{2}{*}{Vertex angles} & \multirow{2}{*}{strongly acute} & one right, & one strongly obtuse, \\
    & & the others strongly acute & the others strongly acute \\
    \bottomrule
\end{tabular}
\caption{Three classes of orthocentric simplices and their properties.}
\label{tab:orttable}
\end{table}

\section{Main lemma} \label{sec:mainlemma}

The following notation will be used throughout this paper. 
Sequences of vectors in $\mathcal{H}$ are denoted by uppercase Greek letters (e.g., $\Phi, \Psi$), while their individual elements are represented by corresponding lowercase letters (e.g., $\phi_j, \psi_j$). Similarly, scalar sequences $s, c, p$ consist of elements $s_j, c_j, p_j$. The notation $s\Phi$ refers to the sequence of vectors formed by the element-wise multiplication $s_j\phi_j$.

\begin{lemma} \label{le:mainlemma}
    Let $\Phi = (\phi_{j})_{j=1}^{d+1} \subset \mathcal{H}$, $p \in \Delta^{\circ}_{d+1}$, and let the \textit{closure condition} $\sum\limits_{j=1}^{d+1} p_{j}\phi_{j} = 0$ hold, i.e., let $p$ be the set of barycentric coordinates of $0$ with respect to $\conv \Phi$. 
    The following statements are equivalent:
    \begin{enumerate}[(a)]
        \item $\sqrt{p}\Phi$ is a tight frame; 
        \item $\conv \Phi$ is an acute orthocentric $d$-simplex with the orthocentre at $0$;
         \item $\conv \Phi$ is \textit{homothetically self-dual}, i.e., homothetic to its dual set.
    \end{enumerate}
    Moreover, if these conditions are satisfied, the Gram matrix of $\Phi$ is given by
    \begin{equation} \label{eq:phigram}
        \langle \phi_{j}, \phi_{k} \rangle = A \left( \frac{\delta_{jk}}{\sqrt{p_{j}p_{k}}} - 1 \right),
    \end{equation}
    for $j,k = 1,\dots,d+1$, where $\delta_{jk}$ is the Kronecker delta symbol, and $A$ represents the frame bound in (a), the negated obtuseness in (b), and the homothety ratio in (c).
\end{lemma}

\begin{proof}    
    Let us assume that $\sqrt{p} \Phi$ is a tight frame with the frame bound $A>0$ and the closure condition holds. If $\phi_{j} = 0$ for some $j$, the remaining vectors of $\Phi$ form an orthogonal basis, because a tight frame consisting of $d$ vectors in a $d$-dimensional real inner product space is necessarily an orthogonal basis of vectors of equal length.
    Thus, $\conv \Phi$ is a rectangular orthocentric simplex with the orthocentre at $0$. This contradicts the assumption $0 \in \interior (\conv \Phi)$, which follows from the closure condition, as the orthocentre of a rectangular orthocentric simplex lies on the boundary. 
    Consequently, we can assume that $\phi_{j} \neq 0$ for every $j$. If $(\phi_{j})_{j=1}^{d}$ spanned a subspace of dimension less than $d$, the closure condition $\sum_{j=1}^{d+1} p_{j}\phi_{j} = 0$ would imply that $\phi_{d+1}$ also belongs to this subspace, contradicting the fact that $\sqrt{p}\Phi$ is a frame and must span the $d$-dimensional space $\mathcal{H}$. Thus, $(\phi_{j})_{j=1}^{d}$ is a basis and the closure condition gives the unique decomposition of $\phi_{d+1}$ with respect to this basis,
    \[ p_{d+1}\phi_{d+1} = - \sum\limits_{j=1}^{d}p_{j}\phi_{j}. \]
    For any $f \in \mathcal{H}$, we obtain
    \begin{align} 
    \nonumber
    f &= \frac{1}{A} \sum\limits_{j=1}^{d+1} p_{j} \langle f,\phi_{j} \rangle \phi_{j} 
    \\ \nonumber
    &= \frac{1}{A} \sum\limits_{j=1}^{d} (p_{j}\langle f,\phi_{j} \rangle - p_{j} \langle f,\phi_{d+1} \rangle) \phi_{j}
    \\
    &= \frac{1}{A} \sum\limits_{j=1}^{d} \langle f,\phi_{j}-\phi_{d+1} \rangle p_{j}\phi_{j}, \label{eq:scalcalculations}
    \end{align}
    so it follows from the uniqueness of decomposition that for any $k,j \in \{1,\dots,d\}$, $k \neq j$ and $f = \phi_{k}$,   
    \begin{equation*}
        \langle \phi_{k}, \phi_{j} \rangle = \langle \phi_{k}, \phi_{d+1} \rangle.
    \end{equation*}
    Since we can repeat this argumentation for any $\phi_{l}$ instead of $\phi_{d+1}$, using Theorem~\ref{th:orthocentric}(b) we get that the simplex $\conv \Phi$ is orthocentric with the orthocentre at $0$. 
    Moreover, the closure condition together with $p_{j} > 0$ guarantees that the simplex is acute orthocentric.

    Conversely, let us assume that $\conv \Phi$ is an acute orthocentric simplex with the orthocentre at $0$ and that $p_{j}$ are the barycentric coordinates of the orthocentre, so the closure condition holds. 
    The orthocentre $0$ lies in the interior of $\conv \Phi$, hence $p \in \Delta^{\circ}_{d+1}$ and $\sigma \coloneqq -\langle \phi_{j}, \phi_{k} \rangle > 0$ for $j \neq k$. 
    The norms of $\phi_{j}$ are given by eq.~\eqref{eq:vertexnorm} as  
    \begin{equation} \label{eq:pvnorm} 
        \lVert \phi_{j} \rVert^{2} = \sigma \left(\frac{1}{p_{j}} - 1 \right). 
    \end{equation}
    
    We use the variational characterisation (Theorem~\ref{th:variationalcharacterisation}) to prove that $\sqrt{p} \Phi$ is a tight frame:
    \begin{align*}
        v(\sqrt{p}\Phi) &= \sum\limits_{j=1}^{d+1} \left( p_{j}^{2} \lVert \phi_{j} \rVert^{4} + \sum\limits_{k \neq j} p_{j}p_{k}\langle\phi_{j},\phi_{k}\rangle^2 \right) - \frac{1}{d}\left( \sum\limits_{j=1}^{d+1}p_{j} \lVert \phi_{j} \rVert^{2} \right)^{2} \\ &= 
        \sum\limits_{j=1}^{d+1} \left( (p_{j}-1)^{2}\sigma^{2} + \sigma^{2}p_{j}\sum\limits_{k \neq j} p_{k} \right) - \frac{\sigma^{2}}{d}\left( \sum\limits_{j=1}^{d+1}(p_{j}-1) \right)^{2} \\ &= 
        \sigma^{2} \sum\limits_{j=1}^{d+1} ( (p_{j}-1)^{2}+p_{j}(1-p_{j}) ) - \frac{\sigma^{2}}{d} \left( \sum\limits_{j=1}^{d+1}p_{j} - (d+1) \right)^{2} \\ &=
        \sigma^{2} \sum\limits_{j=1}^{d+1} ( (p_{j}-1)(p_{j}-1-p_{j}) ) - \frac{\sigma^{2}}{d} \left( \sum\limits_{j=1}^{d+1}p_{j} - (d+1) \right)^{2} \\ &=
        \sigma^{2} \sum\limits_{j=1}^{d+1}(1-p_{j}) - \frac{\sigma^{2}}{d} (1-(d+1))^{2} \\ &=
        \sigma^{2}(d+1-1) - \frac{\sigma^{2}}{d}d^{2} = 0.
    \end{align*}

    Now, let us assume that the equivalent conditions (a)--(b) are satisfied.
    Taking $f = \phi_{k}$ in eq.~\eqref{eq:scalcalculations} and considering the $j=k$ summand, we obtain from the uniqueness of the decomposition
    \begin{equation} 
        \lVert \phi_{k} \rVert^{2} = \frac{A}{p_{k}} + \langle \phi_{k},\phi_{d+1} \rangle = \frac{A}{p_{k}} - \sigma.
    \end{equation}
    Compared to eq.~\eqref{eq:pvnorm}, this condition leads to $\sigma = A$ and the entries of the Gram matrix of $\Phi$ take the form
    \begin{equation*} 
        \langle \phi_{j}, \phi_{k} \rangle = A \left( \frac{\delta_{jk}}{p_{j}} - 1 \right) = A \left( \frac{\delta_{jk}}{\sqrt{p_{j}p_{k}}} - 1 \right)
    \end{equation*}
    as desired.

  Finally, we prove the equivalence of (b) and (c). Let us denote $K \coloneqq \conv \Phi$ and $F_j \coloneqq \conv (\Phi \setminus \{\phi_j\})$ for any $j \in \{1,\dots,d+1\}$, i.e., $F_j$ is the facet of $K$ opposite to $\phi_j$. The dual of $K$ with respect to the sphere with centre at $0$ of radius $\sqrt{A}$ can be represented as \begin{equation*}
        K_{0,\sqrt{A}}^\star=\bigcap_{j=1}^{d+1}H_j^+,
    \end{equation*} where $H_j^+\coloneqq\{x\in\mathcal H:\langle x,\phi_j\rangle\geq -A\}$ and $H_j \coloneqq \{x\in\mathcal H:\langle x,\phi_j\rangle= -A\}$ is a hyperplane perpendicular to $\phi_j$. Note that $K_{0,\sqrt{A}}^\star = A K^\star$, since $K_{0,r}^\star = r^2 K_{0,1}^\star$ for every $r>0$. Since $K$ is a simplex with $0\in\interior K$, so is $K_{0,\sqrt{A}}^\star$, and thus the facets of $K^\star_{0,\sqrt{A}}$ are included in the hyperplanes $H_{j}$, $j = 1,\dots,d+1$. 
        
    Let us assume that $K$ is an acute orthocentric $d$-simplex with the orthocentre at $0$, then there exists $A>0$ such that $\langle \phi_j, \phi_k \rangle = - A$ for $j \neq k$ and consequently $\phi_k \in H_j$ for $k \neq j$.
    This means that $F_j \subset H_j$, and hence $\aff F_j = H_j$, for every $j =1,\dots,d+1$; therefore $K = K_{0,\sqrt{A}}^\star = AK^\star$, as any simplex is determined by the hyperplanes spanned by its facets.

    Conversely, let us assume that $K = AK^\star = K_{0,\sqrt{A}}^\star$ for some $A>0$ and let us fix $j \in \{1,\dots,d+1\}$.
    The facet $F_j$ of the simplex $K$ is a facet of $K^\star_{0,\sqrt{A}}$ as well, so it must be included in one of the bounding hyperplanes, i.e., $F_j \subset H_k$ for some $k \in \{1,\dots,d+1\}$. If $k \neq j$, then $\phi_k \in F_j \subset H_k$, which would imply $\lVert \phi_k \rVert^2 = \langle \phi_k, \phi_k \rangle = -A < 0$, a contradiction. Therefore, we must have $F_j \subset H_j$.
    Consequently, $\phi_k \in H_j$ and $\langle \phi_k,\phi_j \rangle = -A$ for $k \neq j$.
    Therefore, the simplex $K = \conv \Phi$ is orthocentric with the orthocentre at $0$ according to Theorem~\ref{th:orthocentric}(b) and it is acute orthocentric, because $\sigma = A > 0$ (Proposition~\ref{pr:barsigns}).
\end{proof}

\section{Main results} \label{sec:mainresults}

The lemma proven in the previous section applied to GGPTs allows us to thoroughly investigate the relationship between $s$-tight MIC measurements on the one hand, and acute orthocentric simplices centred at the origin (or homothetically self-dual simplices) on the other hand.

\subsection{Orthocentric simplex from minimal \texorpdfstring{$s$}{s}-tight IC measurement}

The proportionality relation will be denoted by $\sim$ throughout this paper.
 
\begin{theorem} \label{th:equiv}
    Let $(\Psi,c)$ be a minimal IC measurement. The following statements are equivalent:
    \begin{enumerate}[(a)]
        \item $(\Psi,c)$ is $s$-tight;
        \item there exists $p \in \Delta^{\circ}_{d+1}$, such that $\conv (\frac{c}{p}\Psi)$ is an acute orthocentric simplex with the orthocentre at $0$;
        \item there exists $p \in \Delta^{\circ}_{d+1}$, such that $\conv (\frac{c}{p}\Psi)$ is homothetic to its dual set.
    \end{enumerate}
    Moreover, if these conditions are satisfied, then $p$ is uniquely defined by
    \begin{equation} \label{eq:pcs}
        p_{j} = (\frac{c_{j}}{s_{j}})^{2} / \sum\limits_{k=1}^{d+1} (\frac{c_{k}}{s_{k}})^{2},    
    \end{equation}
     where $s = (s_{j})_{j=1}^{d+1}$ is any sequence of the scalability constants for $\Psi$
    and the entries of the Gram matrix of $\Psi$ take the form
    \begin{equation} \label{eq:psigram}
    \langle \psi_j, \psi_k \rangle = A \frac{p_j p_k}{c_j c_k} \left( \frac{\delta_{jk}}{\sqrt{p_{j}p_{k}}} - 1 \right),
    \end{equation}
    where $j,k=1,\dots,d+1$, $\delta_{jk}$ denotes the Kronecker delta symbol, $-A$ is the obtuseness of the acute orthocentric simplex $\conv (\frac{c}{p} \Psi)$ in (b), and $A$ is the scale factor in (c).
\end{theorem}
\begin{remark} \label{re:unique_p}
   It follows from eq.~\eqref{eq:psigram} that the probabilities $(p_j)_{j=1}^{d+1}$ satisfy the following conditions:
    \begin{equation*}
        p_j p_k = -\frac{\langle c_j\psi_j,c_k\psi_k\rangle}{A}    
    \end{equation*}
    for $j \neq k$, and
        \begin{equation} \label{eq:probabilities}
        p_j(1-p_j) = \frac{\lVert c_j \psi_j \rVert^2}{A}
    \end{equation}
    for $j=1,\dots,d+1$.
\end{remark}
 \begin{proof}[Proof of the Theorem~\ref{th:equiv}]
    The equivalence of (b) and (c) follows from the equivalence of (b) and (c) in Lemma~\ref{le:mainlemma}.
    
    For the implication (a)$\Rightarrow$(b) let us assume that $(\Psi,c)$ is an $s$-tight IC measurement with scalability constants $s = (s_{j})_{j=1}^{d+1}$, i.e., $s\Psi$ is a tight frame.     
    Let us define $\Phi \coloneqq \frac{s^{2}}{c}\Psi$ and $p_{j} \coloneqq (\frac{c_{j}}{s_{j}})^{2} / \sum\limits_{k=1}^{d+1} (\frac{c_{k}}{s_{k}})^{2}$ for $j=1,\dots,d+1$; then $p \coloneqq (p_{j})_{j=1}^{d+1} \in \Delta^{\circ}_{d+1}$. 
    The frame $\sqrt{p}\Phi \sim s \Psi$ is a tight frame and
    \[ \sum\limits_{j=1}^{d+1} p_{j} \phi_{j} = \frac{\sum\limits_{j=1}^{d+1} c_{j} \psi_{j}}{\sum\limits_{k=1}^{d+1} \left(\frac{c_{k}}{s_{k}}\right)^{2}} = 0. \]
    From implication (a)$\Rightarrow$(b) in Lemma~\ref{le:mainlemma} it follows that $\conv \Phi$ is an acute orthocentric simplex with the orthocentre at $0$.
    Consequently, $\conv (\frac{c}{p} \Psi)$ is an acute orthocentric simplex with the orthocentre at $0$, as $\Phi \sim \frac{c}{p} \Psi$.  
   
    Conversely, to prove (b)$\Rightarrow$(a), let us assume that $\conv (\frac{c}{p} \Psi)$ is an acute orthocentric simplex with the orthocentre at $0$.
    Let us define $\Phi \coloneqq \frac{c}{p} \Psi$, $s \coloneqq \frac{c}{\sqrt{p}}$; then
    \[ \sum\limits_{j=1}^{d+1} p_{j} \phi_{j} = \sum\limits_{j=1}^{d+1} c_{j} \psi_{j} = 0 \]
    and using implication (b)$\Rightarrow$(a) in Lemma~\ref{le:mainlemma}, we obtain that $\sqrt{p}\Phi = s\Psi$ is a tight frame. 
    In order to prove the ``moreover'' part, let us first recall that if $\Psi$ is a scalable frame, then $s$ is uniquely determined by $\Psi$ up to rescaling by a common factor (see Remark~\ref{re:scalability}). In consequence, if $p,p'\in\Delta_{d+1}^\circ$ are such that both $\conv(\frac{c}{p}\Psi)$ and $\conv(\frac{c}{p'}\Psi)$  are acute orthocentric simplices with the orthocentre at 0, the scalability constants for $\Psi$ introduced in the proof above, $s=\frac{c}{\sqrt{p}}$ and $s'=\frac{c}{\sqrt{p'}}$, may differ only by a scalar factor. Since both $p$ and $p'$ have coordinates that sum up to 1, they need to coincide. The final formula \eqref{eq:pcs} follows from the construction of $p$ in the first part of the proof. 
    
    Finally, the entries of the Gram matrix of $\Psi$ are determined by eq.~\eqref{eq:phigram} for $\phi_{j} = \frac{c_j}{p_j}\psi_{j}$.
\end{proof}

\begin{corollary} \label{co:scal_ort}
    If $(\Psi,c)$ is an $s$-tight MIC with scalability constants $s$, then $\conv (\frac{s^{2}}{c}\Psi)$ is an acute orthocentric simplex with the orthocentre at $0$.
\end{corollary}
\begin{proof}
    This follows by the same argument as in the proof of implication (a)$\Rightarrow$(b) in Theorem~\ref{th:equiv} and from eq.~\eqref{eq:pcs}.
\end{proof}

Figure~\ref{fig:sub1} illustrates this rescaling procedure for triangles. Every triangle is orthocentric, although not necessarily with the orthocentre at $0$. Let us observe that the orthocentre of the initial triangle $\conv \Psi$ does not coincide with the origin. However, the triangle $\conv ( \frac{s^{2}}{c} \Psi )$, obtained by appropriately rescaling the vertices of the triangle $\conv \Psi$, is acute orthocentric with the orthocentre at $0$. The altitudes of the transformed triangle intersect at the origin and the points $\psi_{j}$ lie on them.
Since the scalability constants $s_{j}$ are determined up to multiplication, the whole purple triangle can be rescaled arbitrarily. 
Figure~\ref{fig:sub2} shows the analogous construction for tetrahedra (3-simplices), which in general need not be orthocentric  --  the altitudes of the initial tetrahedron are not concurrent. After rescaling the vertices appropriately, we obtain a tetrahedron which is acute orthocentric with the orthocentre at $0$. Again, this tetrahedron is determined up to the scale.

\begin{figure}[H]
\begin{center}
\begin{subfigure}{.48\textwidth}
  \centering
  \includegraphics[width=1\linewidth]{e_triangles.png}
  \caption{Triangles}
  \label{fig:sub1}
\end{subfigure}
\begin{subfigure}{.48\textwidth}
  \centering
  \includegraphics[width=1\linewidth]{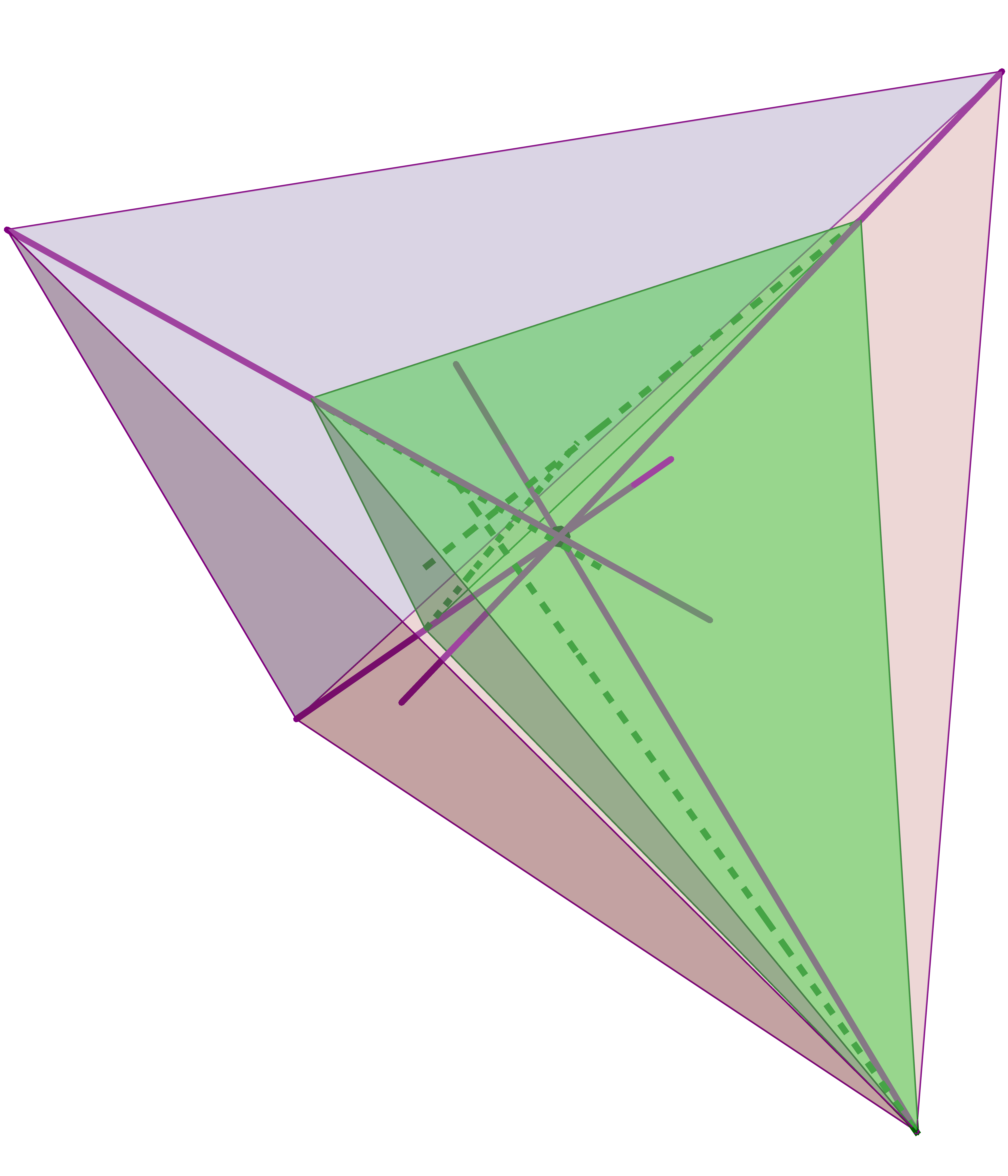}
  \caption{Tetrahedra}
  \label{fig:sub2}
\end{subfigure}
\end{center}
\caption{The vertices of the green triangle $\conv \Psi$, whose orthocentre does not coincide with the origin, can be rescaled to obtain the purple triangle $\conv( \frac{s^{2}}{c}\Psi )$ with the orthocentre at $0$. The altitudes of the initial (green) triangle are represented with dashed green lines, while the altitudes of the transformed (purple) triangle are solid purple lines  --  as can be seen, they intersect at the origin, and the vertices of both triangles lie on them.
The same procedure can be applied to tetrahedra: the altitudes of the initial (green) tetrahedron, represented with dashed green lines, are not concurrent, but the (purple) tetrahedron with rescaled vertices, whose altitudes are represented with solid purple lines, is orthocentric with the orthocentre at $0$. The black point represents the origin in both cases.}
\label{fig:rescaling}
\end{figure}

Theorem~\ref{th:equiv} and eq.~\eqref{eq:psigram} can be applied to known examples of $s$-tight IC measurements such as morphophoric and tight IC measurements, for which the corresponding acute orthocentric simplices take a particularly simple form. 
Namely, for tight IC measurements the scaling is trivial, whereas in the morphophoric case, the rescaling procedure leads to a regular simplex. 
Moreover, in both cases the probabilities $p$ do not appear in the Gram matrix of $\Psi$.

\begin{corollary} \label{co:tightic}
    Let $(\Psi,c)$ be a minimal IC measurement. The following conditions are equivalent:
    \begin{enumerate}[(a)]
        \item $(\Psi,c)$ is a tight IC measurement;
        \item $\conv \Psi$ is an acute orthocentric simplex with the orthocentre at $0$.
    \end{enumerate}
    If these conditions are satisfied, the entries of the Gram matrix of $\Psi$ take the form
    \begin{equation*}
        \langle \psi_j, \psi_k \rangle = A \left( \frac{\delta_{jk}}{\sqrt{c_{j}c_{k}}} - 1 \right).  
    \end{equation*}
\end{corollary}
\begin{proof}
    Let us assume that $(\Psi,c)$ is a tight IC measurement, i.e., $\Psi$ is a scalable frame with scalability constants $s_j \sim \sqrt{c_j}$, then $p_j = c_j$ according to eq.~\eqref{eq:pcs} and using implication (a)$\Rightarrow$(b) in Theorem~\ref{th:equiv} we obtain that the simplex $\conv (\frac{c}{p} \Psi) = \conv \Psi$ is acute orthocentric with the orthocentre at $0$. 
    Conversely, if $\conv \Psi$ is an acute orthocentric simplex with the orthocentre at $0$, then it follows from implication (b)$\Rightarrow$(a) in Theorem~\ref{th:equiv} with $p_j = c_j$ that $\Psi$ is a scalable frame. According to eq.~\eqref{eq:pcs}, the scalability constants satisfy $s_j \sim \sqrt{c_j}$, hence $(\Psi,c)$ is a tight IC measurement.
    The formula for the entries of the Gram matrix follows from eq.~\eqref{eq:psigram}.
\end{proof}

\begin{corollary} \label{co:morphophoric}
    Let $(\Psi,c)$ be a minimal IC measurement. The following conditions are equivalent:    
    \begin{enumerate}[(a)]
        \item $(\Psi,c)$ is a morphophoric measurement;
        \item $\conv (c \Psi)$ is a regular simplex with the orthocentre at $0$.
    \end{enumerate}
    If these conditions are satisfied, the entries of the Gram matrix of $\Psi$ take the form
    \begin{align*}
        \langle \psi_j, \psi_k \rangle = A \frac{(d+1)\delta_{jk} - 1}{(d+1)^{2} c_j c_k}.
    \end{align*}
\end{corollary}
\begin{proof}
    Let us assume that $(\Psi,c)$ is a morphophoric measurement, i.e., $\Psi$ is a scalable frame with scalability constants $s_j \sim c_j$, then $p_j = 1/(d+1)$ according to eq.~\eqref{eq:pcs} and using implication (a)$\Rightarrow$(b) in Theorem~\ref{th:equiv} we obtain that $\conv (c\Psi)$ is an acute orthocentric simplex with the orthocentre at $0$. The centroid of $\conv (c\Psi)$ coincides with the orthocentre, so this simplex is regular \cite[Theorem 4.1]{edmonds2005orthocentric}.
    Conversely, if $\conv (c\Psi)$ is a regular simplex with the orthocentre at $0$, then it is necessarily acute orthocentric. It follows from implication (b)$\Rightarrow$(a) in Theorem~\ref{th:equiv} with $p_j = 1/(d+1)$ that $\Psi$ is a scalable frame. According to eq.~\eqref{eq:pcs}, the scalability constants satisfy $s_j \sim c_j$, hence $(\Psi,c)$ is a morphophoric measurement.
    The formula for the entries of the Gram matrix follows from eq.~\eqref{eq:psigram}.
\end{proof}

\subsection{Minimal \texorpdfstring{$s$}{s}-tight IC measurements from orthocentric simplex} \label{sec:converse}

Corollary~\ref{co:scal_ort} establishes that an $s$-tight MIC measurement determines an acute orthocentric simplex. It is, therefore, a natural question whether an acute orthocentric simplex, under some additional assumptions, determines an $s$-tight IC measurement. 
A simplex is defined only by the set of its vertices, while to specify a measurement, both $\Psi$ and $c$ are needed. In the case of minimal IC measurements, $c$ is unique for a given $\Psi$. However, in Corollary~\ref{co:scal_ort}, the vertices of the obtained orthocentric simplex are rescaled vectors of $\Psi$, $\phi_{j} = \frac{s_{j}^{2}}{c_{j}} \psi_{j}$.
Therefore, the obtained measurement should be of the form $(\Psi,c)$, where elements of $\Psi$ are the vertices of the orthocentric simplex multiplied by factors involving $c$ and $s$. As a consequence, $c$ might not be uniquely determined by the set of vertices. 
It turns out that the converse of Corollary~\ref{co:scal_ort} is true in the following sense.

\begin{theorem} \label{th:ort_scal}
    If $\Phi \subset V_0$ is the set of vertices of an acute orthocentric simplex with the orthocentre at $0$, then for any $c \in \Delta^{\circ}_{d+1}$ there exists $s=(s_{j})_{j=1}^{d+1}$ with $s_{j} > 0$, such that  $(\frac{c}{s^{2}}\Phi,c)$ is a minimal $s$-tight IC measurement with the scalability constants $s$.
\end{theorem}
\begin{proof}
    The orthocentre of an acute orthocentric simplex lies in the interior, so $0 \in \interior (\conv \Phi)$ and there exists a unique $p \in \Delta^{\circ}_{d+1}$ such that $\sum\limits_{j=1}^{d+1} p_{j}\phi_{j} = 0$.  
    Let us take $\gamma_{j} \coloneqq C\sqrt{p_{j}}$ for some $C > 0$ (to be determined later).
    Now, let us fix arbitrary $c \in \Delta^{\circ}_{d+1}$ and let us define
    \begin{equation*}
        s_{j} \coloneqq \frac{c_{j}}{\gamma_{j}} = \frac{c_{j}}{C\sqrt{p_{j}}},
    \end{equation*}
    then 
    \begin{equation*}
        \sum \limits_{j=1}^{d+1} \frac{c_{j}^{2}}{s_{j}^{2}}\phi_{j} = 0. 
    \end{equation*}  
    By implication (b)$\Rightarrow$(a) in Lemma~\ref{le:mainlemma}, $\sqrt{p}\Phi$ is a tight frame, so $\frac{c}{s}\Phi$ is also a tight frame for any $C > 0$. 
    It remains to choose $C > 0$ such that $\frac{c_{j}}{s_{j}^{2}}\phi_{j} = \frac{C^{2}p_{j}}{c_{j}}\phi_{j} \in S^\star$ for every $j=1,\dots,d+1$, which is possible for all sufficiently small $C$, since $0 \in \interior S^\star$.
\end{proof}

\begin{remark} 
In Corollary~\ref{co:scal_ort}, $\Psi$ determines $c$ (uniquely) and $s$ (uniquely up to rescaling by a common factor), then $(\Psi,c)$ and $s$ determine the orthocentric simplex. In Theorem~\ref{th:ort_scal}, the orthocentric simplex $\conv \Phi$ and arbitrary $c \in \Delta^{\circ}_{d+1}$ determine $s$ uniquely up to rescaling by a common factor, provided $\frac{c_{j}}{s_{j}^{2}}\phi_{j} \in S^\star$ for every $j=1,\dots,d+1$.
\end{remark}

In what follows, $p \in \Delta^{\circ}_{d+1}$ denotes the barycentric coordinates of $0$ with respect to $\conv \Phi$ and $C>0$ the constant from the proof of Theorem~\ref{th:ort_scal}; rescaling $s$ by a common factor is equivalent to rescaling $C$ by the inverse factor.
Similarly to the case of Theorem~\ref{th:equiv}, we will examine how the special cases of $s$-tight IC measurements, namely morphophoric and tight IC measurements, appear in Theorem~\ref{th:ort_scal}. 
For morphophoric and tight IC measurements, the scalability constants are proportional to $c_{j}$ and $\sqrt{c_{j}}$, respectively. Changing this factor of proportionality (which means varying $s$) is equivalent to changing $C$. In particular, it can be used to guarantee that $\frac{c}{s^{2}}\Phi \subset S^\star$ and we obtain a measurement. 

\begin{corollary} \label{co:morphophoric_c}
     The measurement obtained in Theorem~\ref{th:ort_scal} is morphophoric if and only if the simplex $\conv \Phi$ is a regular simplex with the orthocentre at $0$.     
\end{corollary}
\begin{proof}
    See the proof of Corollary~\ref{co:morphophoric}.
\end{proof}

\begin{corollary} \label{co:tightic_c}
    For every $\Phi$ as in Theorem~\ref{th:ort_scal}, there is exactly one $c \in \Delta^{\circ}_{d+1}$ that gives a tight IC measurement.    
\end{corollary}
\begin{proof}
    The obtained measurement is tight IC if and only if $c_j = p_j$, as in the proof of Corollary~\ref{co:tightic}.    
\end{proof}

In Theorem~\ref{th:ort_scal}, an acute orthocentric simplex generates a class of minimal $s$-tight IC measurements. The next remark lists several properties of such class, in particular results concerning morphophoric and tight IC measurements discussed above, as well as the role of $S^\star$.   

\begin{remark} \label{re:class}
\hfill
\begin{enumerate}[(a)]
    \item For any $\Phi$ and $c \in \Delta^{\circ}_{d+1}$, each measurement that can be obtained from these $\Phi$ and $c$ corresponds to a value of $C > 0$ satisfying $\frac{C^{2}p_{j}}{c_{j}}\phi_{j} \in S^\star$ for $j=1,\dots,d+1$, where $p_{j}$ are uniquely determined by $\Phi$. The set $S^\star$ is a neighbourhood of $0$, so the set of allowed values of $C$ is of the form $(0, C_{\max}(\Phi,c,S^\star)]$, where the upper bound depends on $\Phi$, $c$ and $S^\star$ and can be written as
    \begin{equation*}
        C_{\max}(\Phi,c,S^\star) = \max \left\{ C>0 \;\big|\; \frac{C^{2}p_{j}}{c_{j}} \phi_{j} \in S^\star \text{ for } j=1,\dots,d+1\right\}.
    \end{equation*}
    The maximum is attained as $S^\star$ is compact.
    \item For any $\Phi$ and $c \in \Delta^{\circ}_{d+1}$, the possible values of $s$ are given by $s_{j} = \frac{c_{j}}{C\sqrt{p_j}}$, where $p_{j}$ are uniquely determined by $\Phi$ and $C \in (0,C_{\max}(\Phi,c,S^\star)]$ as above. In particular, the set of allowed values of each $s_{j}$ is bounded from below.
    \item For any $\Phi$, the corresponding class of measurements contains exactly one tight IC measurement, up to rescaling of vectors $\frac{c_{j}}{s_{j}^{2}}\phi_{j}$ by a common factor.
    This unique tight IC measurement corresponds to $c_{j}=p_{j}$.
    The possibility of rescaling these vectors by a common factor stems from the fact that $C$ is not uniquely determined by $\Phi$ and $c$, it is only required to satisfy $C^{2}\phi_{j} \in S^\star$ for every $j=1,\dots,d+1$.
    \item All morphophoric measurements which can be obtained in Theorem~\ref{th:ort_scal} are necessarily obtained from a regular simplex. Any measurement obtained from a regular simplex is morphophoric.
    \item If a measurement obtained in Theorem~\ref{th:ort_scal} is both morphophoric and tight IC, then the simplex $\conv \Phi$ is regular (as is required by the morphophoricity of the obtained measurement).
    However, in this case $c$ is unique, because the measurement is tight IC, which implies
    \begin{equation*}
        c_j = p_j = \frac{1}{d+1}.
    \end{equation*}
\end{enumerate}     
\end{remark}

Let us note that the family generated by $\Phi$ depends on the simplex only through the directions of its vertices, along which all the measurement vectors point; indeed, among acute orthocentric simplices with the orthocentre at $0$, the directions determine the vertices up to rescaling by a common factor.
Transforming the directions by a rotation yields a family with identical internal data: the same Gram matrix, angles, and outcome probabilities.
Reflections, too, preserve all the internal data, yet within a fixed GGPT a family and its mirror image need not be physically equivalent. \footnote{The mechanism can be seen in the quantum GGPT of Section~\ref{sec:quantum}, where $\mathfrak H\simeq\mathbb C^{D}$, the space $V_{0}$ consists of the traceless Hermitian operators on $\mathfrak H$, and $d=D^{2}-1$.
By Wigner's theorem the orthogonal symmetries of the quantum state space are the unitary and antiunitary conjugations, $\rho\mapsto U\rho U^{\dagger}$ and $\rho\mapsto U\rho^{T}U^{\dagger}$, forming $\mathrm{Ad}\,U(D)\rtimes\mathbb{Z}_{2}$, the second component being the coset of the transposition $\tau\colon\rho\mapsto\rho^{T}$ \cite{bengtsson2017geometry}. This group has two connected components for every $D$, but their position in $O(d)$ depends on $D$: the subgroup $\mathrm{Ad}\,U(D)$ is connected and hence contained in $SO(d)$, whereas $\tau$ negates precisely the $D(D-1)/2$ imaginary Hermitian directions, so that $\det\bigl(\tau|_{V_{0}}\bigr)=(-1)^{D(D-1)/2}$. 
Thus for $D\equiv 2,3\pmod 4$ the second component lies in $O(d)\setminus SO(d)$ and provides a mirror, so that a family and its reflection always coincide. For $D\equiv 0,1\pmod 4$ -- the smallest instance being the ququart, $D=4$ -- both components lie in $SO(d)$; no symmetry of the state space acts on $V_{0}$ with determinant $-1$, and the two mirror families are genuinely inequivalent.}

\section{\texorpdfstring{$s$}{s}-tight MIC measurements -- algebraic conditions} \label{sec:algebraiccriterion}

Theorem~\ref{th:equiv}, which establishes the fundamental geometric correspondence via a rescaling procedure, gives a criterion for determining whether a MIC measurement is $s$-tight. This criterion takes a particularly simple form for tight IC and morphophoric measurements.

\begin{corollary} \label{co:crit}
    A MIC measurement $(\Psi,c)$ is $s$-tight if and only if there exist $p \in \Delta^{\circ}_{d+1}$ and $A>0$ such that 
    \begin{equation} \label{eq:criterion}
        \langle \frac{c_j}{p_j} \psi_j, \frac{c_k}{p_k} \psi_k \rangle = -A
    \end{equation}
    for every $j,k = 1,\dots,d+1$, $j \neq k$. 
    Moreover, a MIC measurement $(\Psi,c)$ is
    \begin{enumerate}[(a)]
        \item tight IC if and only if $\langle \psi_j, \psi_k \rangle = -A$,
        \item morphophoric if and only if $\langle c_j\psi_j, c_k\psi_k \rangle = -A$,
    \end{enumerate}
    for every $j,k = 1,\dots,d+1$, $j \neq k$, and some $A > 0$.
\end{corollary}
\begin{proof}
    Follows from the equivalence of (a) and (b) in Theorem~\ref{th:equiv}. 
    The tight IC and morphophoric cases follow from Corollary~\ref{co:tightic} and Corollary~\ref{co:morphophoric}, respectively.
\end{proof}

However, one might ask if $s$-tightness can be diagnosed directly from the measurement vectors $\Psi$, without computing the barycentric coordinates $c$, the probabilities $p$ or the scalability constants $s$. Criteria of this kind are available for scalable frames in general, and the following theorem rests on the one obtained by Kutyniok et al.~\cite[Corollary 2.9]{kutyniok2013scalable}, which characterises
scalability of $d+1$ vectors through the quantities
$\langle\psi_i,\psi_k\rangle\langle\psi_k,\psi_j\rangle/\langle\psi_i,\psi_j\rangle$. What we add is a reformulation adapted to the present setting: the criterion is recast as two conditions that are projective in nature -- an obtuse angle condition and a cross-ratio rule -- depending on the directions of the measurement vectors alone. In this form it is the algebraic counterpart of the orthocentricity of the underlying simplex, and it is this form that yields the classification of
Section~\ref{sec:skeletons}.

\needspace{5\baselineskip}
\begin{theorem} 
\label{th:s_tight_criterion}
Let $(\Psi,c)$ be an IC measurement with $d+1$ outcomes. Then the following statements are equivalent:
\begin{enumerate}[(a)]
\item $(\Psi,c)$ is $s$-tight;
\item
\begin{enumerate}[(I)]
\item (\textit{obtuse angle condition})
$$\langle \psi_j, \psi_k \rangle < 0 \quad \text{for } j \neq k, \ j,k = 1,\dots,d+1;$$ 
\item (\textit{cross-ratio rule}) 
$$\langle \psi_i, \psi_j \rangle \langle \psi_k, \psi_l \rangle = \langle \psi_i, \psi_l \rangle \langle \psi_k, \psi_j \rangle \quad \text{for pairwise distinct } i, j, k, l = 1,\dots,d+1.$$
\end{enumerate}
\end{enumerate}
\end{theorem}

\begin{proof} \hfill 

(a)$\Rightarrow$(b) Let us assume that $(\Psi,c)$ is an $s$-tight MIC. It follows from implication (a)$\Rightarrow$(b) in Theorem~\ref{th:equiv} that the simplex $\conv (\frac{c}{p}\Psi)$ is acute orthocentric with the orthocentre at $0$ for some $p \in \Delta^{\circ}_{d+1}$, so there exists $A>0$ such that $\langle \frac{c_j}{p_j}\psi_{j},\frac{c_k}{p_k}\psi_{k}\rangle = -A$ for $j \neq k$. Consequently, $\langle \psi_j,\psi_k \rangle = -A \frac{p_j p_k}{c_j c_k} < 0$ for $j \neq k$ and 
\begin{align*}
    \frac{\langle \psi_i, \psi_j \rangle \langle \psi_k, \psi_l \rangle}{\langle \psi_i, \psi_l \rangle 
\langle \psi_k, \psi_j \rangle} = 1
\end{align*}
 for pairwise distinct $i, j, k, l = 1,\dots,d+1$. 

(b)$\Rightarrow$(a) The obtuse angle condition guarantees that none of the inner products $\langle\psi_i,\psi_j\rangle$ is equal to zero and thus, by \cite[Corollary 2.9]{kutyniok2013scalable}, it suffices to show that for all $k\in\{1,\dots,d+1\}$ and for all pairwise distinct indices $i, j, k$
\begin{equation} \label{eq:Kutyniokcondition}
    \frac{\langle\psi_i,\psi_k\rangle\langle\psi_k,\psi_j\rangle}{\langle\psi_i,\psi_j\rangle}=\textnormal{const}(k)<0.
\end{equation}
Under condition (I) the cross-ratio rule can be equivalently expressed as \begin{equation*}
    \frac{\langle\psi_k,\psi_l\rangle}{\langle\psi_i,\psi_l\rangle}= \frac{\langle\psi_k,\psi_j\rangle}{\langle\psi_i,\psi_j\rangle} \textnormal{ for pairwise distinct } i, j, k, l = 1,\dots,d+1.
\end{equation*}
By multiplying these expressions on both sides by $\langle\psi_i,\psi_k\rangle$ and using the symmetry of scalar product we deduce that the left-hand side in \eqref{eq:Kutyniokcondition} depends only on $k$. The negative sign of this constant follows from the obtuse angle condition.
\end{proof}

\begin{remark} \label{re:MIC}
For $d=2$ the cross-ratio rule~(II) is vacuous, as there are no four pairwise distinct indices in $\{1,2,3\}$; hence, in this case, a MIC is $s$-tight if and only if the obtuse angle condition~(I) holds. 
This reflects the fact that every triangle is orthocentric, cf. Section~\ref{sec:skeletons}. 
\end{remark}

\section{Examples} \label{sec:examples}

\subsection{Disk}

Let us consider the two-dimensional case of ball GGPT \cite[Example D]{szymusiak2025can}. 
The set of states is the unit disk $S \coloneqq \{x \in \mathbb{R}^2: \lVert x \rVert \leq 1\}$ and this GGPT is self-dual, i.e., $S^\star = S$. 
A measurement is described by a pair $(\Psi,c) = (\psi_j,c_j)_{j=1}^{n}$, where $\Psi \subset S$ and $c \in \Delta_n^\circ$ satisfy the closure condition, $\sum\limits_{j=1}^{n} c_j \psi_j = 0$. 
Geometrically, a measurement is represented as a convex polytope included in the disk $S$, with the origin in its interior.
In this representation, $\Psi$ is the set of vertices of the polytope and $c$ is the set of barycentric coordinates of the origin with respect to $\Psi$. 
If the measurement is a MIC, this convex polytope is a triangle in $S$ and $c$ is uniquely determined by $\Psi$.

\subsection{Boxworld}

The Boxworld (or gbit, ``generalised bit''; see~\cite{barrett2007information, plavala2023general, janotta2011limits}), is the regular polygonal GGPT \cite[Example C]{szymusiak2025can} for $N=4$.
Let $u_k = (\cos(k\pi / 2), \sin(k \pi / 2))$ for $k=1,\dots,4$ be the vertices of a square in $\mathbb{R}^2$.
This GGPT is defined by $V_0 \coloneqq \mathbb{R}^2$ (with the standard inner product) and $S \coloneqq \conv (a u_k)_{k=1}^4$, for some $a > 0$, which is a square. 
The dual of $S$ is $S^\star = \{ (x,y) \in \mathbb{R}^2 : |x| \leq 1/a, |y| \leq 1/a \}$. 
It is also a square, but rotated by $\pi/4$ and rescaled by the factor $\sqrt{2}a^{-2}$, so this GGPT is not self-dual.
A measurement $(\Psi,c)$ is represented as a convex polytope $\conv \Psi$, included in the square $S^\star$, such that $0 \in \interior (\conv \Psi)$, and $c$ is the set of barycentric coordinates of the origin with respect to $\Psi$. Similarly as in the previous example, in the case of MICs this polytope is a triangle and $c$ is uniquely determined by $\Psi$.

\subsection{Quantum}
\label{sec:quantum}
A quantum GGPT \cite[Example B]{szymusiak2025can} with $\mathfrak{H} \simeq \mathbb{C}^{D}$ is defined as follows: $V_0 \coloneqq \{\rho \in \mathcal{L}_{s}(\mathfrak{H}): \tr \rho = 0\}$, where $\mathcal{L}_s(\mathfrak{H})$ denotes the space of linear self-adjoint (i.e., Hermitian) operators on $\mathfrak{H}$, equipped with the Hilbert--Schmidt inner product, $\langle \rho_1, \rho_2 \rangle_{HS} \coloneqq \tr (\rho_1 \rho_2)$ for $\rho_1, \rho_2 \in \mathcal{L}_{s}(\mathfrak{H})$, the inner product on $V_0$ is the restriction of the Hilbert--Schmidt inner product to $V_0$, and the set of states is represented as
\begin{equation*}
  S \coloneqq \{ \sqrt{D} (\rho - I/D): \rho \in \mathcal{L}_s(\mathfrak{H}), \rho \geq 0, \tr \rho = 1\}.
\end{equation*}
Here, $\rho$ is a density matrix and $m \coloneqq I/D$ is the maximally mixed state.
The quantum GGPT is self-dual, $S^\star = S$, and its real dimension is $d = D^2-1$; for the geometry of this state space see
\cite{bengtsson2013geometry,bengtsson2017geometry}.

A measurement is given by a POVM $\Pi = (\Pi_{j})_{j=1}^{n}$, where $\Pi_{j} \in \mathcal{L}_{s}(\mathfrak{H})$, $\Pi_{j} \geq 0$ and $\sum\limits_{j=1}^{n} \Pi_{j} = I$. 
In the language of a quantum GGPT, the measurement $\Pi$ is represented as $(\Psi,c) = (\psi_j, c_j)_{j=1}^{n}$ with
\begin{align*}
    c_{j} = \Pi_{j}(m) = \frac{1}{D} \tr (\Pi_{j}), \quad \psi_{j} = \frac{P_{0}(\Pi_{j})}{\sqrt{D} c_j} = \frac{\sqrt{D}}{\tr(\Pi_j)} P_{0}(\Pi_{j}),
\end{align*}
where $P_0: \mathcal{L}_{s}(\mathfrak{H}) \to V_0$ is the orthogonal projection onto $V_0$, given by $P_0(\rho) = \rho - \tr (\rho) m$ for $\rho \in \mathcal{L}_{s}(\mathfrak{H})$.
A~MIC measurement consists of $n = d+1 = D^{2}$ effects such that $\lin \Psi = \lin (P_0(\Pi_{j}))_{j=1}^{D^2} = V_0$. 
\begin{remark}
A similar representation of POVMs was used in~\cite{debrota2021varieties}. 
Namely, each effect $\Pi_j$ was written in the form
\[ \Pi_j = e_j \rho_j, \quad \text{where } e_j = \tr(\Pi_j),\;\rho_j = \frac{1}{\tr(\Pi_j)}\Pi_j, \]
which implies $\tr(\rho_j)=1$. The correspondence between these two representations is given by
\[ c_j = \frac{e_j}{D}, \quad \psi_j = \sqrt{D}\rho_j - \frac{1}{\sqrt{D}}I. \]
The entries of the Gram matrices of $\Psi$ and $(\rho_j)_{j=1}^{n}$ are related by
\[ \langle \psi_j, \psi_k \rangle_{HS} = D \tr(\rho_j \rho_k) - 1 = D \langle \rho_j,\rho_k \rangle_{HS} - 1.  \]
\end{remark}

According to Corollary~\ref{co:crit}, a given MIC $(\Psi,c)$ is $s$-tight IC if and only if there exist $A > 0$ and $p \in \Delta_{d+1}^\circ$ such that for $j \neq k$,
\begin{equation*} 
    \langle \frac{c_j}{p_j} \psi_{j}, \frac{c_k}{p_k} \psi_{k} \rangle_{HS} = -A.
\end{equation*}
For quantum GGPT, these conditions take the form
\begin{align} \label{eq:quantum2}
    \langle P_{0}(\Pi_{j}), P_{0}(\Pi_{k}) \rangle_{HS} = -D A p_j p_k.
\end{align}

\subsection{Qubit}

Let us consider the qubit case of quantum GGPT, $\mathfrak{H} = \mathbb{C}^{2}$, then the real dimension is $d=3$. 
Each element of $S$ is of the form $\sqrt{2}\rho - I/\sqrt{2}$ for some density matrix $\rho$, and $\rho$ can be written in the Bloch representation, 
\begin{equation*}
    \rho = \frac{1}{2} (I + \vec{b} \cdot \vec{\sigma}), 
\end{equation*}
where $\vec{b} \in \mathbb{R}^3$ and $\vec{\sigma}$ denotes the vector of unnormalised Pauli matrices, i.e., $\lVert\sigma_{x}\rVert_{HS} = \lVert\sigma_{y}\rVert_{HS} = \lVert\sigma_{z}\rVert_{HS} = \sqrt{2}$. 
Positive semi-definiteness of $\rho$ is equivalent to $\lVert \vec{b} \rVert \leq 1$ and 
\begin{equation*}
    \sqrt{2}\rho - I/\sqrt{2} = \frac{1}{\sqrt{2}} \vec{b} \cdot \vec{\sigma},
\end{equation*}
therefore the set of states $S$ is the Bloch ball of radius $1$; indeed, $\vec b\mapsto\frac{1}{\sqrt2}\,\vec b\cdot\vec\sigma$ is an isometric isomorphism of $\mathbb{R}^3$ onto $V_0$, since $\tr \bigl((\vec a\cdot\vec\sigma)(\vec b\cdot\vec\sigma)\bigr)=2\,\vec a\cdot\vec b$.
The Bloch representation can be extended to POVMs. 
Namely, each effect of a POVM $\Pi = (\Pi_j)_{j=1}^n$ can be represented as $\Pi_j = c_j I + \vec{b}_j \cdot \vec{\sigma}$ and the conditions $\Pi_j \geq 0$ and $\sum\limits_{j=1}^{n} \Pi_{j} = I$ take the form
\begin{align} \label{eq:assumptions}
    \lVert \vec{b}_j \rVert \leq c_j, \quad \sum\limits_{j=1}^{n} c_j = 1, \quad \sum\limits_{j=1}^{n} \vec{b}_j = 0.
\end{align}
The projections of the effects $\Pi_{j}$ onto $V_{0}$ are $P_{0}(\Pi_{j}) = \vec{b}_{j} \cdot \vec{\sigma}$, hence the POVM $\Pi$ can be described by $(\Psi,c)$, where 
\begin{align*}
    c_j = \Pi_j(m) = \frac{\tr (\Pi_j)}{2}, \quad \psi_j = \frac{\sqrt{2}}{\tr(\Pi_j)} \vec{b}_{j} \cdot \vec{\sigma}.  
\end{align*}
Substituting the trace yields 
\begin{equation*}
    \psi_j = \frac{1}{\sqrt{2} c_j} \vec{b}_j \cdot \vec{\sigma}
\end{equation*}
and calculating the squared Hilbert--Schmidt norm of $\psi_j$, we obtain
\begin{equation*}
    \lVert \psi_j \rVert_{HS}^2 = \tr (\psi_j^2) = \frac{1}{2c_j^2} \tr \left((\lVert \vec{b}_j \rVert^2 I)\right) = \frac{1}{2c_j^2} \cdot 2\lVert \vec{b}_j \rVert^2 = \frac{\lVert \vec{b}_j \rVert^2}{c_j^2},
\end{equation*}
which reveals a simple proportionality
\begin{equation} \label{eq:psibc}
    \lVert \psi_j \rVert_{HS} = \frac{\lVert \vec{b}_j \rVert}{c_j}.
\end{equation} 

A MIC consists of $n=4$ effects and can be represented as a tetrahedron $\conv \Psi$ included in the Bloch ball $S$ such that $0 \in \interior (\conv \Psi)$ and $c$ are the barycentric coordinates of the origin, determined uniquely by this tetrahedron.

The criterion \eqref{eq:quantum2} for $s$-tightness of a MIC $(\Psi,c)$ takes the form
\begin{equation} \label{eq:qubit}
    \vec{b}_j \cdot \vec{b}_k = -A p_j p_k,
\end{equation}
for $j \neq k$, where $\vec{b}_j \cdot \vec{b}_k$ denotes the standard inner product in $\mathbb{R}^{3}$ and conditions \eqref{eq:assumptions} must be satisfied. 
Eqs. \eqref{eq:qubit} mean that the convex hull of $(1/p_j) \vec{b}_j$, $j = 1,\dots,4$, is an acute orthocentric simplex with the orthocentre at the origin.
Equivalently, a MIC $(\Psi,c)$ is $s$-tight with scalability constants $s$ if and only if the tetrahedron in $\mathbb{R}^3$ whose vertices are $(s_j^2 / c_j^2) \vec{b}_j$, $j=1,\dots,4$, is acute orthocentric with the orthocentre at $0$.

Conditions \eqref{eq:qubit} take simpler forms in the case of morphophoric or tight IC measurements. 
Let us recall that a measurement is morphophoric if and only if it is $s$-tight IC with the scalability constants $s_{j} \sim c_{j}$ and tight IC if and only if it is $s$-tight IC with the scalability constants $s_{j} \sim \sqrt{c_{j}}$. 
If the measurement is a MIC, the former corresponds to $p_j = 1/(d+1)$ and the latter to $p_j = c_j$.

\begin{corollary}
    In the qubit case, a MIC is morphophoric if and only if
    \begin{equation*}
        \vec{b}_j \cdot \vec{b}_k = \text{const}
    \end{equation*}
    and the constant is negative.
    These equations are satisfied if and only if the tetrahedron in $\mathbb{R}^3$ with vertices $\vec{b}_j$, $j=1,\dots,4$, is acute orthocentric with the orthocentre at $0$.
\end{corollary}

\begin{corollary}
   In the qubit case, a MIC is tight IC if and only if
   \begin{equation*}
       \frac{\vec{b}_j \cdot \vec{b}_k}{c_j c_k} = \text{const}
    \end{equation*}
    and the constant is negative.
    These equations are satisfied if and only if the tetrahedron in $\mathbb{R}^3$ with vertices $(1/c_j) \vec{b}_j$, $j=1,\dots,4$, is acute orthocentric with the orthocentre at $0$.
\end{corollary}

\subsubsection{Axially symmetric configuration (one-parameter subfamily)} \label{sec:axial}

To gain a deeper understanding of the permissible direction space in Theorem~\ref{th:s_tight_criterion}, it is instructive to analyse an example with high geometric symmetry.
Let us consider a one-parameter family of MIC measurements in the qubit case, such that the measurement vectors $\psi_j$, $j=1,\dots,4$, lie on the unit sphere. We further assume that one of the measurement vectors $(\psi_j)_{j=1}^4$ is fixed at the North Pole of the unit sphere $\mathbb{S}^2$,
\begin{equation*}
    \psi_1 = \begin{pmatrix} 0 \\ 0 \\ 1 \end{pmatrix},
\end{equation*}
and that the remaining three measurement vectors $\psi_2, \psi_3, \psi_4$ lie on the same latitude at a constant height $z_0$, 
\begin{equation*}
    \psi_k = \begin{pmatrix} \sqrt{1-z_0^2} \cos\theta_k \\ \sqrt{1-z_0^2} \sin\theta_k \\ z_0 \end{pmatrix}, \quad \text{for } k \in \{2,3,4\},
\end{equation*}
where $\theta_k$ denote the respective azimuthal angles on the sphere.

Applying the obtuse angle condition (Condition I in Theorem~\ref{th:s_tight_criterion}) to the pairs of measurement vectors containing the pole yields
\begin{equation*}
    \langle \psi_1, \psi_k \rangle = z_0 < 0, \quad \text{for } k \in \{2,3,4\},
\end{equation*}
which means that the analysed latitude $z_0$ must lie strictly on the lower hemisphere. Now, the cross-ratio rule (Condition II in Theorem~\ref{th:s_tight_criterion}) takes the form
\begin{align*}
    z_0 \langle \psi_3, \psi_4 \rangle = z_0 \langle \psi_2, \psi_4 \rangle = z_0 \langle \psi_2, \psi_3 \rangle.
\end{align*}
The above equations reduce to a relation of equiangularity between the measurement vectors lying on the latitude,
\begin{align*}
    \langle \psi_2, \psi_3 \rangle = \langle \psi_2, \psi_4 \rangle = \langle \psi_3, \psi_4 \rangle = q
\end{align*}
for some $q < 0$ (negativity of $q$ follows from the obtuse angle condition), which implies that the measurement vectors $\psi_2, \psi_3, \psi_4$ define the vertices of an equilateral triangle. 
Without loss of generality, we can make a rigid choice of azimuthal phases: $\theta_2 = 0$, $\theta_3 = \frac{2\pi}{3}$, $\theta_4 = \frac{4\pi}{3}$.
The inner product $q$ is related to $z_0$ by a polynomial relation,
\begin{align*}
    q = (1-z_0^2)\cos\left(\frac{2\pi}{3}\right) + z_0^2 = -\frac{1}{2}(1-z_0^2) + z_0^2 = \frac{3}{2}z_0^2 - \frac{1}{2}
\end{align*}
and the requirement $q < 0$ gives the final restriction on the value of $z_0$,
\begin{align*}
    \frac{3}{2}z_0^2 - \frac{1}{2} < 0 \implies z_0 \in \left(-\frac{1}{\sqrt{3}}, 0\right) \approx (-0.577, 0).
\end{align*} 

The barycentric coordinates of the origin with respect to $\Psi$, forming a probability distribution $c$, are uniquely determined by the closure condition and the normalisation as functions of $z_0$, 
\begin{align*}
    c_1(z_0) &= \frac{-z_0}{1 - z_0} = \frac{|z_0|}{1 + |z_0|}, \\
    c_k(z_0) &= \frac{1}{3(1 - z_0)} = \frac{1}{3(1 + |z_0|)}, \quad \text{for } k \in \{2,3,4\}.
\end{align*}

Within the determined range of possible structures, there exists a unique configuration of maximal symmetry, in which the mutual inner product of all four measurement vectors $\psi_1,\dots,\psi_4$ becomes identical, i.e., $z_0 = q$. 
This solution occurs for $z_0 = -1/3$, which corresponds to the situation where the configuration of directions generates a regular tetrahedron inscribed in the sphere.
This case represents the structure with the highest degree of symmetry, namely the symmetric informationally complete (SIC) measurement.

The geometric constraint that the measurement vectors $\psi_j$ lie on the unit sphere, $\lVert \psi_j \rVert_{HS}=1$, together with eq.~\eqref{eq:psibc}, implies $\lVert \vec{b}_j \rVert = c_j$.
Since the eigenvalues of the effect $\Pi_j$ are given by $c_j \pm \lVert \vec{b}_j \rVert$, this equality guarantees that the smallest eigenvalue is always exactly zero. 
Consequently, all POVMs in this one-parameter family are composed entirely of rank-1 effects (proportional to pure states).
Applying the derived relation $\lVert \vec{b}_j \rVert = c_j$, the Bloch vectors are parametrised as:
\begin{align}
    \vec{b}_1 &= \begin{pmatrix} 0 \\ 0 \\ \frac{-z_0}{1-z_0} \end{pmatrix}, \\
    \vec{b}_k &= \frac{1}{3(1-z_0)} \begin{pmatrix} \sqrt{1-z_0^2}\cos\theta_k \\ \sqrt{1-z_0^2}\sin\theta_k \\ z_0 \end{pmatrix}, \quad \text{for } k \in \{2,3,4\}.
\end{align}

As the parameter $z_0$ varies within the permissible bounds $z_0 \in (-\frac{1}{\sqrt{3}}, 0)$, the measurement changes accordingly:
\begin{itemize}
    \item \textbf{Flat limit ($z_0 \to 0^-$):} The latitude approaches the equator, the measurement vectors $\psi_2, \psi_3, \psi_4$ approach the equatorial plane. The Bloch vector $\vec{b}_1$ continuously vanishes, as $\lVert \vec{b}_1 \rVert=c_1 \to 0$ and thus the measurement vector $\psi_1$ no longer takes active part in the measurement process. In consequence, the measurement loses its informational completeness along the Z-axis.
    \item \textbf{SIC-POVM ($z_0 = -1/3$):} The point of highest symmetry: substituting $z_0 = -1/3$, we obtain a SIC-POVM -- namely, $\conv \Psi$ is a regular tetrahedron inscribed in the unit sphere. The probability distribution given by the barycentric coordinates of $0$ is unbiased ($c_j = 1/4$) and the Bloch vectors have equal lengths ($\lVert \vec{b}_j \rVert = 1/4$).
    \item \textbf{Obtuse angle limit ($z_0 \to -1/\sqrt{3}$):} The value of $c_1$ approaches its maximum ($c_1 \to \frac{1}{1+\sqrt{3}} \approx 0.366$), measurement vectors $\psi_2, \psi_3, \psi_4$ approach mutual orthogonality, because
    \[ \langle \psi_2, \psi_3 \rangle = \langle \psi_2, \psi_4 \rangle = \langle \psi_3, \psi_4 \rangle = \frac{3}{2}z_0^2 - \frac{1}{2} \to 0. \]
    The lengths of corresponding Bloch vectors $\vec{b}_2, \vec{b}_3, \vec{b}_4$, equal to $c_2,c_3,c_4$ respectively, shrink as $c_1$ increases. 
\end{itemize}

Figure~\ref{fig:three_examples} presents three examples of $s$-tight MICs belonging to the axially symmetric family, represented as the sets of measurement vectors $\Psi$ in the unit sphere. They correspond to different values of the parameter $z_0$: $z_0 = -0.1$ (close to the upper bound), $z_0 = -1/3$ (maximal symmetry, SIC-POVM) and $z_0 = -0.55$ (close to the lower bound). 
For $z_0 = -0.1$, the measurement vectors $\psi_2$, $\psi_3$, $\psi_4$ are close to the equatorial plane and almost perpendicular to $\psi_1$. 
In the case of SIC-POVM, $\conv \Psi$ is a regular tetrahedron inscribed in the unit sphere. 
In the third example, with $z_0 = -0.55$ (close to the obtuse angle limit), the measurement vectors $\psi_2$, $\psi_3$, $\psi_4$ are almost mutually orthogonal and the angles between $\psi_1$ and $\psi_k$, $k=2,3,4$, are almost as large as possible for an $s$-tight MIC measurement belonging to this family.

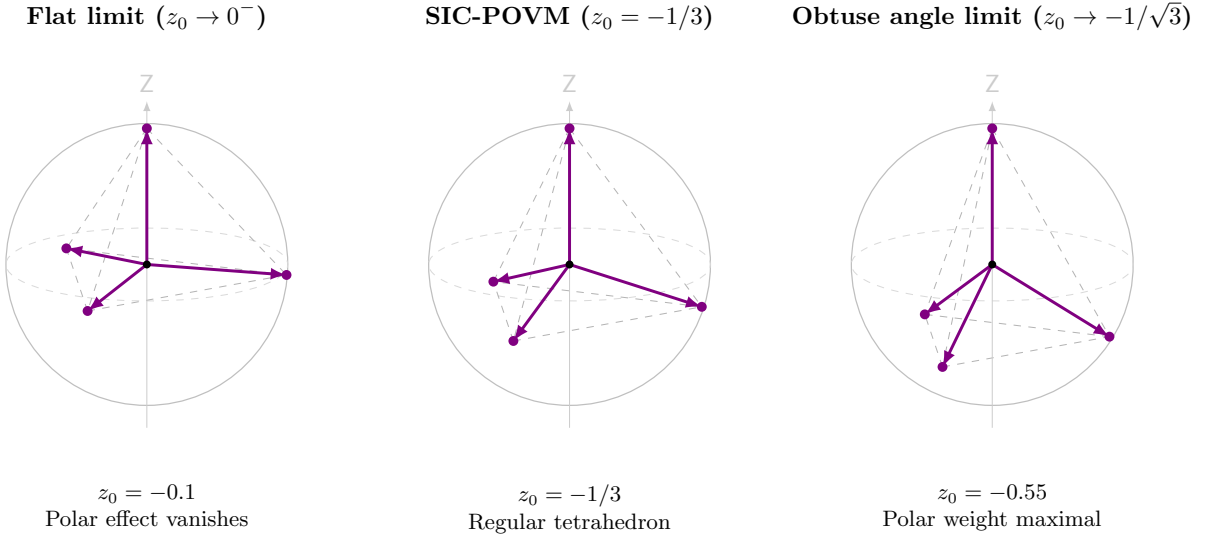
\begin{figure}[H] 
\begin{center}
\resizebox{\linewidth}{!}{%
\tdplotsetmaincoords{75}{115} 

\begin{tikzpicture}[font=\sffamily]

\begin{scope}[xshift=0cm]
    \node[anchor=south, font=\bfseries] at (0, 3.2) {Flat limit ($z_0 \to 0^-$)};
    \node[anchor=north, align=center, font=\small, text=black] at (0, -3.0) {$z_0 = -0.1$ \\ Polar effect vanishes};

    \draw[gray!50, line width=0.5pt] (0,0) circle (2);

    \begin{scope}[tdplot_main_coords, scale=2]
        \coordinate (O) at (0,0,0);
        \coordinate (P1) at (0, 0, 1);
        \coordinate (P2) at (0.995, 0, -0.1);
        \coordinate (P3) at (-0.497, 0.862, -0.1);
        \coordinate (P4) at (-0.497, -0.862, -0.1);

        \tdplotdrawarc[gray!30, dashed]{(0,0,0)}{1}{0}{360}{}{}

        \draw[black!30, dashed] (P1) -- (P2); \draw[black!30, dashed] (P1) -- (P3); \draw[black!30, dashed] (P1) -- (P4);
        \draw[black!30, dashed] (P2) -- (P3); \draw[black!30, dashed] (P2) -- (P4); \draw[black!30, dashed] (P3) -- (P4);

        \draw[-latex, gray!40] (0,0,-1.2) -- (0,0,1.2) node[anchor=south]{Z};

        \draw[-latex, violet, line width=1.2pt] (O) -- (P1);
        \draw[-latex, violet, line width=1.2pt] (O) -- (P2);
        \draw[-latex, violet, line width=1.2pt] (O) -- (P3);
        \draw[-latex, violet, line width=1.2pt] (O) -- (P4);

        \fill[violet] (P1) circle (1pt); \fill[violet] (P2) circle (1pt);
        \fill[violet] (P3) circle (1pt); \fill[violet] (P4) circle (1pt);
        \fill[black] (O) circle (0.8pt);
    \end{scope}
\end{scope}

\begin{scope}[xshift=6cm]
    \node[anchor=south, font=\bfseries] at (0, 3.2) {SIC-POVM ($z_0 = -1/3$)};
    \node[anchor=north, align=center, font=\small, text=black] at (0, -3.0) {$z_0 = -1/3$ \\ Regular tetrahedron};

    \draw[gray!50, line width=0.5pt] (0,0) circle (2);

    \begin{scope}[tdplot_main_coords, scale=2]
        \coordinate (O) at (0,0,0);
        \coordinate (S1) at (0, 0, 1);
        \coordinate (S2) at (0.943, 0, -0.333);
        \coordinate (S3) at (-0.471, 0.816, -0.333);
        \coordinate (S4) at (-0.471, -0.816, -0.333);

        \tdplotdrawarc[gray!30, dashed]{(0,0,0)}{1}{0}{360}{}{}
        
        \draw[black!30, dashed] (S1) -- (S2); \draw[black!30, dashed] (S1) -- (S3); \draw[black!30, dashed] (S1) -- (S4);
        \draw[black!30, dashed] (S2) -- (S3); \draw[black!30, dashed] (S2) -- (S4); \draw[black!30, dashed] (S3) -- (S4);

        \draw[-latex, gray!40] (0,0,-1.2) -- (0,0,1.2) node[anchor=south]{Z};

        \draw[-latex, violet, line width=1.2pt] (O) -- (S1);
        \draw[-latex, violet, line width=1.2pt] (O) -- (S2);
        \draw[-latex, violet, line width=1.2pt] (O) -- (S3);
        \draw[-latex, violet, line width=1.2pt] (O) -- (S4);

        \fill[violet] (S1) circle (1pt); \fill[violet] (S2) circle (1pt);
        \fill[violet] (S3) circle (1pt); \fill[violet] (S4) circle (1pt);
        \fill[black] (O) circle (0.8pt);
    \end{scope}
\end{scope}

\begin{scope}[xshift=12cm]
    \node[anchor=south, font=\bfseries] at (0, 3.2) {Obtuse angle limit ($z_0 \to -1/\sqrt{3}$)};
    \node[anchor=north, align=center, font=\small, text=black] at (0, -3.0) {$z_0 = -0.55$ \\ Polar weight maximal};

    \draw[gray!50, line width=0.5pt] (0,0) circle (2);

    \begin{scope}[tdplot_main_coords, scale=2]
        \coordinate (O) at (0,0,0);
        \coordinate (O1) at (0, 0, 1);
        \coordinate (O2) at (0.835, 0, -0.55);
        \coordinate (O3) at (-0.418, 0.723, -0.55);
        \coordinate (O4) at (-0.418, -0.723, -0.55);

        \tdplotdrawarc[gray!30, dashed]{(0,0,0)}{1}{0}{360}{}{}
        
        \draw[black!30, dashed] (O1) -- (O2); \draw[black!30, dashed] (O1) -- (O3); \draw[black!30, dashed] (O1) -- (O4);
        \draw[black!30, dashed] (O2) -- (O3); \draw[black!30, dashed] (O2) -- (O4); \draw[black!30, dashed] (O3) -- (O4);

        \draw[-latex, gray!40] (0,0,-1.2) -- (0,0,1.2) node[anchor=south]{Z};

        \draw[-latex, violet, line width=1.2pt] (O) -- (O1);
        \draw[-latex, violet, line width=1.2pt] (O) -- (O2);
        \draw[-latex, violet, line width=1.2pt] (O) -- (O3);
        \draw[-latex, violet, line width=1.2pt] (O) -- (O4);

        \fill[violet] (O1) circle (1pt); \fill[violet] (O2) circle (1pt);
        \fill[violet] (O3) circle (1pt); \fill[violet] (O4) circle (1pt);
        \fill[black] (O) circle (0.8pt);
    \end{scope}
\end{scope}

\end{tikzpicture}
}%
\caption{
    Three examples of axially symmetric measurements, represented as the tetrahedra $\conv \Psi$ for $z_0 = -0.1$ (close to the flat limit), $z_0 = -1/3$ (SIC-POVM) and $z_0 = -0.55$ (close to the obtuse angle limit).
}
\label{fig:three_examples}
\end{center}
\end{figure}

Table~\ref{tab:symmetric_subfamily} presents the parameters of the analysed one-parameter subfamily of $s$-tight MIC measurements for the qubit, depending on the value of the structural parameter $z_0 \in (-\frac{1}{\sqrt{3}},0)$. 
Each measurement is described by a pair $(\Psi,c)$, where $c \in \Delta_{d+1}^\circ$ is uniquely determined by $\Psi$. 
The unique probability distribution $p$ such that the tetrahedron $\conv (\frac{c}{p}\Psi)$ is acute orthocentric with the orthocentre at $0$ is given by the formula \eqref{eq:probabilities} and $A>0$ is the negated obtuseness of this orthocentric tetrahedron, as in Theorem~\ref{th:equiv}.

The upper bound on the value of $z_0$ follows from the obtuse angle condition for the pairs of measurement vectors containing $\psi_1$ and the informational completeness of the measurement. In the flat limit, the measurement vectors $\psi_2, \psi_3, \psi_4$ approach the equatorial plane and the volume of the sub-tetrahedron $\conv (\{0,\psi_2,\psi_3,\psi_4\})$ vanishes, hence $c_1 \to 0$.
The configuration of maximal symmetry corresponds to $z_0 = -1/3$ and the MIC measurement $(\Psi,c)$ being a SIC-POVM. In this case, the simplex $\conv \Psi$ is a regular tetrahedron inscribed in the unit sphere and the measurement is unbiased, $c_j = 1/4$ for $j=1,\dots,4$.
The lower bound on the value of $z_0$ follows from the obtuse angle condition for $\psi_k$, $k=2,3,4$. Namely, as $z_0 \to -1/\sqrt{3}$, these vectors approach mutual orthogonality, so that $\langle \psi_j,\psi_k\rangle \to 0$ for $j,k \in \{2,3,4\}$, $j\neq k$. In this limit $p_1 \to 1$ and $p_k \to 0$ for $k=2,3,4$. Note that, under the normalisation $\lVert \psi_j \rVert_{HS}=1$ adopted here, the remaining vertices $\frac{c_k}{p_k}\psi_k$, $k=2,3,4$, recede to infinity and $A$ diverges. The simplex $\conv(\frac{c}{p}\Psi)$ degenerates to an unbounded trirectangular tetrahedron. 

\begin{table}[H]
\centering
\begin{tabular}{Sl Sc Sc Sc Sc}
\toprule

\multirow{3}{*}{\textbf{Parameter}} &
\multirow{2}{*}{\textbf{Flat limit}} &
\multirow{2}{*}{\textbf{SIC-POVM}} &
\textbf{Obtuse angle} & 
\multirow{2}{*}{\textbf{General formula}} \\

&
& 
&
\textbf{limit} &
\\

&
$z_0 \to 0^-$ &
$z_0 = -1/3$ &
$z_0 \to -1/\sqrt{3}$ &
$z_0 \in (-1/\sqrt{3}, 0)$ \\ 
\midrule

$\psi_1$ & 
$\begin{pmatrix} 0 \\ 0 \\ 1 \end{pmatrix}$ & 
$\begin{pmatrix} 0 \\ 0 \\ 1 \end{pmatrix}$ & 
$\begin{pmatrix} 0 \\ 0 \\ 1 \end{pmatrix}$ & 
$\begin{pmatrix} 0 \\ 0 \\ 1 \end{pmatrix}$ \\ 

$\psi_k$ ($k=2,3,4$) & 
$\begin{pmatrix} \cos\theta_k \\ \sin\theta_k \\ 0 \end{pmatrix}$ & 
$\begin{pmatrix} \frac{\sqrt{8}}{3}\cos\theta_k \\ \frac{\sqrt{8}}{3}\sin\theta_k \\ -1/3 \end{pmatrix}$ & 
$\begin{pmatrix} \frac{\sqrt{2}}{\sqrt{3}}\cos\theta_k \\ \frac{\sqrt{2}}{\sqrt{3}}\sin\theta_k \\ -1/\sqrt{3} \end{pmatrix}$ & 
$\begin{pmatrix} \sqrt{1-z_0^2}\cos\theta_k \\ \sqrt{1-z_0^2}\sin\theta_k \\ z_0 \end{pmatrix}$ \\ 
\midrule

$c_1$ & $0$ & $1/4$ & $\frac{1}{1+\sqrt{3}} \approx 0.366$ & $\frac{|z_0|}{1+|z_0|}$ \\
$c_k$ ($k=2,3,4$) & $1/3$ & $1/4$ & $\frac{1}{3+\sqrt{3}} \approx 0.211$ & $\frac{1}{3(1+|z_0|)}$ \\ 
\midrule

$p_1$ & $0$ & $1/4$ & $\to 1$ & $\frac{2z_0^2}{1-z_0^2}$ \\
$p_k$ ($k=2,3,4$) & $1/3$ & $1/4$ & $\to 0$ & $\frac{1-3z_0^2}{3(1-z_0^2)}$ \\ 
\midrule

$A$ & $1/2$ & $1/3$ & $\to +\infty$ & $\frac{(1+z_0)^2}{2(1-3z_0^2)}$ \\ 
\bottomrule
\end{tabular}
\caption{Geometric and physical parameters of the axially symmetric configuration for the qubit.}
\label{tab:symmetric_subfamily}
\end{table}

\begin{remark} \label{re:curran}
The axially symmetric configurations coincide exactly with the skewed SIC measurements introduced by Curran~\cite{curran2026quantum} for the qubit ($D=2$), albeit parametrised differently and with a crucial restriction on the parameter range.
\end{remark}

\subsection{Equiangular MIC measurements} \label{sec:equiangular}

The class of equiangular MIC measurements, introduced in \cite[Section 3.4]{debrota2021varieties} for the quantum case, can be generalised to an arbitrary GGPT in a natural way. In the standard language we would say that the measurement $(g_j)_{j=1}^{d+1}$ is an equiangular MIC if there exist $\alpha,\zeta >0$ such that
\begin{equation} 
    \langle g_i, g_j \rangle = \alpha \delta_{ij} + \zeta
\end{equation}
for $i,j=1,\dots,d+1$. In the new approach this definition translates as follows: a MIC measurement $(\Psi,c)$ is called \textit{equiangular} if there exist $\beta,\xi >0$ such that
\begin{equation} \label{eq:gequiangular}
    \langle \psi_i, \psi_j \rangle_0 = \frac{\beta \delta_{ij} + \xi}{c_ic_j}-1
\end{equation}
for $i,j=1,\dots,d+1$. In order to see the equivalence of these definitions it suffices to note how these inner products are related if we put $g_j = T( \frac{c_j}{\sqrt{\mu}} \psi_j + \frac{c_j}{\mu}m )$ as in Proposition~\ref{pr:effect}:
\begin{equation*} 
\begin{aligned} 
    \langle g_i, g_j \rangle &= \left\langle T( \frac{c_i}{\sqrt{\mu}} \psi_i + \frac{c_i}{\mu}m ), T( \frac{c_j}{\sqrt{\mu}} \psi_j + \frac{c_j}{\mu}m ) \right\rangle = \frac{c_i c_j}{\mu} \left\langle \psi_i + \frac{1}{\sqrt{\mu}}m, \psi_j + \frac{1}{\sqrt{\mu}}m \right\rangle \\ &= \frac{c_i c_j}{\mu} \left( \langle \psi_i, \psi_j \rangle_0 + \frac{\lVert m \rVert^2}{\mu} \right) = \frac{c_i c_j}{\mu} \left( \langle \psi_i, \psi_j \rangle_0 + 1 \right).
\end{aligned}
\end{equation*}
In particular, the constants from the definitions are related in the following way: $\beta=\alpha\mu$ and $\xi=\zeta\mu$.

An equiangular MIC is necessarily unbiased. Indeed, from the closure condition we have
\begin{equation*}
   0=\left\langle\sum_{i=1}^{d+1}c_i\psi_i,c_j\psi_j\right\rangle_0 =\sum_{i=1}^{d+1}(\beta\delta_{ij}+\xi-c_ic_j)=\beta+(d+1)\xi-c_j
\end{equation*}
and thus $c_j$ does not depend on $j$, so that $c_j=1/(d+1)$. Consequently, $\xi=1/(d+1)^2-\beta/(d+1)$. Accordingly, eq.~\eqref{eq:gequiangular} takes the form
\begin{equation*}
    \langle \psi_i, \psi_j \rangle_0=(d+1)^2(\beta \delta_{ij}+1/(d+1)^2-\beta/(d+1)) -1=(d+1)^2\beta\delta_{ij}-(d+1)\beta
\end{equation*}
and so, for $i \neq j$ we obtain
\begin{equation} \label{eq:psiequiangular}
    \langle \psi_i, \psi_j \rangle_0 = -(d+1)\beta<0.  
\end{equation} 

   This fact guarantees that the simplex $\conv \Psi$ is acute orthocentric. 
    Finally, it follows from Corollary~\ref{co:crit} that every equiangular MIC measurement is tight IC; being also unbiased, it is therefore morphophoric.

\section{Conclusions} \label{sec:conclusions}

\subsection{Summary} \label{sec:concl-summary}

The operational notion of tightness thus admits a geometric counterpart in classical Euclidean geometry. For a minimal IC measurement $(\Psi,c)$ in an arbitrary GGPT, being \mbox{$s$-tight} is equivalent to the concurrence of the altitudes of a suitably rescaled measurement simplex, with the orthocentre located at the origin of the state space, and, equivalently, to the homothetic self-duality of this simplex. 
The correspondence works in both directions: every acute orthocentric simplex with the orthocentre at the origin arises in this way, generating a family of minimal $s$-tight IC measurements that contains, up to overall rescaling, exactly one tight IC measurement. 
In this dictionary the known special classes acquire transparent geometric signatures: tight IC measurements correspond to trivial rescaling, whereas morphophoric ones are precisely those whose rescaled simplex is regular. Moreover, the class of minimal $s$-tight IC measurements is governed entirely by the measurement directions, through the obtuse angle condition~(I) and the cross-ratio condition~(II).

A configuration studied since the nineteenth century for its intrinsic elegance thus re-emerges as the geometric form of operational tightness. Read from
geometry to physics, the correspondence organises the landscape of
minimal IC measurements and reduces their design to a constructive
procedure: choose an admissible set of directions, fix the target
probabilities, and anchor the resulting simplex in the state space. Read from physics to geometry, it raises the question of which orthocentric simplices can be anchored in a given convex body -- see problem~(ii) below.

\subsection{Skeletons and anchoring} \label{sec:skeletons}

We close by summarising, in condensed form, the complete classification of minimal $s$-tight IC measurements that our results entail.\footnote{Full details of the arguments sketched in this subsection will be given in a forthcoming work.}

\emph{The skeleton.}
By Theorem~\ref{th:equiv} and Remark~\ref{re:scalability}, every minimal $s$-tight IC measurement $(\Psi,c)$ determines a unique vector $p\in\Delta^{\circ}_{d+1}$ -- its \emph{skeleton} -- for which $\conv(\frac{c}{p}\Psi)$ is an acute orthocentric simplex with the orthocentre at $0$; equivalently, $p$ is the vector of barycentric coordinates of the orthocentre of the simplex.
The skeleton encodes the entire angular structure of the measurement: passing to the unit directions \mbox{$\eta_i\coloneqq\psi_i/\lVert\psi_i\rVert$,} $i=1,\dots,d+1$, eq.~\eqref{eq:psigram} gives
\begin{equation} \label{eq:fingerprint}
\langle\eta_i,\eta_j\rangle=-t_it_j \quad (i\neq j), \qquad t_i\coloneqq\sqrt{\frac{p_i}{1-p_i}}.
\end{equation}
Conversely, conditions (I)--(II) of Theorem~\ref{th:s_tight_criterion} force precisely this form: the cross-ratio rule factorises the off-diagonal Gram entries as $-t_it_j$ with $t_k^{2}=-g_{ik}g_{kj}/g_{ij}>0$ (cf.~\eqref{eq:Kutyniokcondition}), and the requirement that $d+1$ unit vectors realising this Gram matrix fit into $\mathbb{R}^{d}$ amounts, by the matrix determinant lemma, exactly to $\sum_{i=1}^{d+1} p_i=1$ for $p_i\coloneqq t_i^{2}/(1+t_i^{2})$. 
Since pairwise-obtuse vectors have the property that any $d$ of them are linearly independent, the sign $\varepsilon \coloneqq \sgn\det(\eta_1,\dots,\eta_d)$ is well defined and invariant under rotations. 
To summarise: modulo rotations, the admissible direction configurations of minimal $s$-tight IC measurements are classified bijectively by the pair $(p,\varepsilon)\in\Delta^{\circ}_{d+1}\times\{\pm1\}$. For each $p$ we fix once and for all a \emph{reference configuration}
$(\eta_1,\dots,\eta_{d+1})$ realising \eqref{eq:fingerprint}; by the
above, every configuration with the same skeleton is then of the form
$g\eta_1,\dots,g\eta_{d+1}$ for some $g\in O(d)$.
The uniform skeleton $p_i\equiv 1/(d+1)$ labels the morphophoric (regular) family, $p=c$ singles out the tight IC measurement within each family, and skeletons with repeated entries correspond to configurations with extra symmetry, such as the axially symmetric family of Section~\ref{sec:axial}. 
In this sense the skeleton determines the structure of the $s$-tight MIC measurements near the white noise POVMs.
Indeed, conditions (I)--(II) are invariant under positive rescaling of the individual vectors, so the admissible tuples form a cone, while $S^{\star}$ contains a ball $B(0,r)$: every tuple of length below $r$ is therefore admissible in one GGPT exactly when it is admissible in another.
Consequently, in a neighbourhood of the white noise POVMs the set of minimal $s$-tight IC measurements is the same in all GGPTs of a given dimension -- determined by $d$ alone, through the moduli space $\Delta^{\circ}_{d+1}\times\{\pm1\}$ of directions and the freely varying lengths -- while the individual state space makes itself felt only globally, by truncating those lengths.

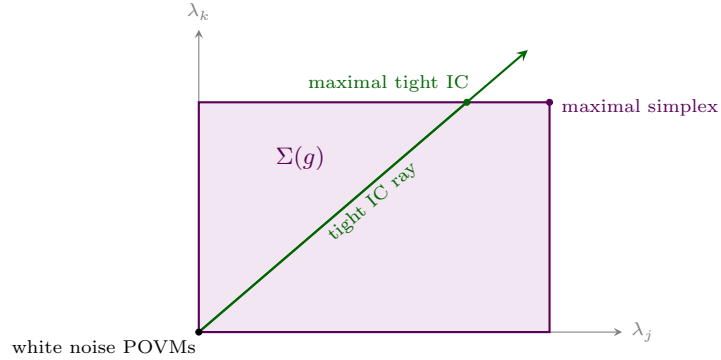
\begin{figure}[H]
\centering
\begin{tikzpicture}[scale=3.2, >=stealth]
\draw[->, gray] (0,0) -- (1.75,0) node[right, font=\scriptsize, text=gray] {$\lambda_j$};
\draw[->, gray] (0,0) -- (0,1.25) node[above, font=\scriptsize, text=gray] {$\lambda_k$};
\fill[violet!10] (0,0) rectangle (1.45,0.95);
\draw[violet!70!black, thick] (0,0) rectangle (1.45,0.95);
\node[violet!70!black, font=\small] at (0.42,0.72) {$\Sigma(g)$};
\draw[->, thick, green!40!black] (0,0) -- (1.36,1.166);
\draw[green!40!black, thick] (0,0) -- (1.108,0.95);
\filldraw[green!40!black] (1.108,0.95) circle (0.35pt);
\node[green!40!black, font=\scriptsize, above left=-2pt] at (1.14,0.97) {maximal tight IC};
\node[green!40!black, font=\scriptsize, rotate=40.6] at (0.72,0.54) {tight IC ray};
\filldraw[violet!70!black] (1.45,0.95) circle (0.35pt);
\node[violet!70!black, font=\scriptsize, right=1pt] at (1.45,0.93) {maximal simplex};
\filldraw (0,0) circle (0.35pt);
\node[font=\scriptsize, below left=-1pt] at (0.02,0) {white noise POVMs};
\end{tikzpicture}
\caption{Anatomy of the family of minimal $s$-tight IC measurements generated by a fixed acute orthocentric simplex (Theorem~\ref{th:ort_scal}), drawn schematically in the lengths $\lambda_1,\ldots,\lambda_{d+1}$ of the measurement vectors (two of the coordinates shown). 
For a fixed anchoring $g$ of the directions the admissible lengths fill a box $\Sigma(g)$ determined by $S^{\star}$ (shaded); the tight IC measurements ($c=p$, Corollary~\ref{co:tightic_c}) form a ray, which leaves the box at the maximal tight IC measurement, while the outer corner marks the maximal simplex, at which all lengths attain their bounds simultaneously; the white noise POVMs sit in the degenerate limit $\lambda\to 0$, which is not itself a member of the family (the limiting measurements are no longer informationally complete).
Only one such box is shown: rotating the configuration yields a continuum of boxes, its mirror reflections yield another, and whether the two continua describe the same physical family is the chirality question addressed in a forthcoming paper.}
\label{fig:box}
\end{figure}

\emph{Anchoring.}
A skeleton is state-space-independent mathematical data; to become a physical measurement it must be anchored in a concrete state space $S$. 
An \emph{anchoring} is a choice of $g\in O(d)$, carrying the reference directions to $g\eta_i$, together with lengths $\lambda_i>0$; the vectors $\psi_i\coloneqq\lambda_i\,g\eta_i$ (with $c$ recovered from $p$ and $\lambda$) form a measurement if and only if $\psi_i\in S^{\star}$, i.e., if each length does not exceed the radial function of the dual body, $\lambda_i\leq 1/h_S(-g\eta_i)$, where $h_S(\eta)\coloneqq\max_{x\in S}\langle x,\eta\rangle$.
Each anchoring thus carries its own box of admissible lengths, $\Sigma(g)\coloneqq\prod_{i}\bigl(0,\,1/h_S(-g\eta_i)\bigr]$, whose shape varies continuously with $g$ through the support function of $S$, and the set of all minimal $s$-tight IC measurements with skeleton $p$ is the bundle of boxes
\begin{equation*}
\mathcal{M}_p\;\cong\;\bigl\{(g,\lambda)\,:\,g\in O(d),\ \lambda\in\Sigma(g)\bigr\},
\end{equation*}
within each box the tight IC measurements forming a single ray, $\lambda_i\propto 1/t_i$ -- the geometric picture behind Theorem~\ref{th:ort_scal}, Remark~\ref{re:class} and Corollary~\ref{co:tightic_c}, illustrated in Figure~\ref{fig:box}. 

Two anchored measurements are physically identical when a symmetry of the theory combined with a relabelling of the outcomes carries one onto the other: \mbox{$(g,\lambda)\sim(u\,g\,h_{\pi},\,\pi\cdot\lambda)$,} where $u$ runs over the orthogonal symmetries of $S$ and $\pi$ over the permutations preserving $p$, implemented by $h_{\pi}\in O(d)$ with $\det h_{\pi}=\sgn\pi$. 
The physical classes of measurements with skeleton $p$ therefore correspond to the points of the double quotient of $O(d)$ by these two groups, fibred by the boxes. 
Altogether, a minimal $s$-tight IC measurement is fully determined by its intrinsic skeleton $p\in\Delta^{\circ}_{d+1}$ together with the class of the pair $(g,\lambda)$ modulo the identifications above -- the first datum independent of the state space, the second recording how the configuration sits inside it.

\subsection{The fiducial perspective} \label{sec:concl-fiducial}

As stated at the outset, tightness and morphophoricity are relative to the Euclidean structure with which the underlying GGPT is endowed.
Indeed, let $(\Psi,c)$ be an arbitrary minimal IC measurement, and let
\begin{equation} \label{eq:linear-parts}
\ell_j \coloneqq c_j\langle\psi_j,\,\cdot\,\rangle_0 = \pi_j - c_j, \qquad j=1,\dots,d+1,
\end{equation}
be the linear parts of its outcome functionals~\eqref{eq:pj}. 
Suppose that we treat this measurement as \emph{fiducial} -- i.e., as a trusted reference standard in the spirit of QBism, where a fixed IC measurement serves as the informational ``coordinate system'' through which all states and all other measurements are expressed \cite{fuchs2013quantum} -- and fix a \emph{calibration}: a target vector $P\in\Delta^{\circ}_{d+1}$ of reference probabilities together with a scale $A>0$.

The \emph{fiducial scalar product} of this calibration,
\begin{equation} \label{eq:fiducial-ip}
\langle x,y\rangle_{P,A} \coloneqq \frac{1}{A}\sum_{j=1}^{d+1}\frac{1}{P_j}\,\ell_j(x)\,\ell_j(y), \qquad x,y\in V_0,
\end{equation}
is then the Euclidean structure that the fiducial measurement itself induces on the state space. 
It depends only on the measurement and the calibration, leaves all outcome probabilities unchanged, and renders the measurement $s$-tight with skeleton $P$. 
With respect to $\langle\cdot,\cdot\rangle_{P,A}$, the Riesz vectors of the functionals \eqref{eq:linear-parts} form, after the rescaling of Theorem~\ref{th:equiv}, an acute orthocentric simplex with the orthocentre at $0$. Moreover, it is the only inner product with this property.

\begin{proposition} \label{pr:fiducial-unique}
Let $(\Psi,c)$ be a minimal IC measurement, $P\in\Delta^{\circ}_{d+1}$ and $A>0$.
If, with respect to an inner product $g$ on $V_0$, the measurement is $s$-tight with skeleton $P$ and frame bound $A$, then $g=\langle\cdot,\cdot\rangle_{P,A}$. 
Consequently, the inner products compatible with a given minimal IC measurement are parametrised, up to an overall scale, bijectively by $\Delta^{\circ}_{d+1}$.
\end{proposition}

\begin{proof}
The frame identity for the rescaled frame of Theorem~\ref{th:equiv} yields $A\,g(x,y)=\sum_{j=1}^{d+1}P_j^{-1}\,\ell_j(x)\,\ell_j(y)$ for all $x,y\in V_0$, whose right-hand side does not refer to $g$ at all; by \eqref{eq:fiducial-ip} this is $A\,\langle x,y\rangle_{P,A}$. The second statement follows, as $P$ ranges over $\Delta^{\circ}_{d+1}$.
\end{proof}

In particular, every minimal IC measurement becomes tight IC under one choice of the inner product ($P=c$) and morphophoric under another ($P$ uniform). Read through Theorem~\ref{th:equiv}, this has a purely geometric consequence:

\begin{corollary} \label{co:every-simplex}
Every $d$-simplex containing the origin in its interior is an acute orthocentric simplex with the orthocentre at $0$ for a suitably chosen Euclidean structure.
\end{corollary}

A systematic development of this fiducial point of view, in which the Euclidean structure is not presupposed but constructed from the reference measurement and its calibration, is the subject of a companion work in preparation~\cite{slomczynski2026fiducial}.

\subsection{Open problems} \label{sec:concl-open}

Several natural problems remain open.
\begin{enumerate}[(i)]
    \item \emph{Beyond minimality.} For measurements with $n > d+1$ outcomes the simplex is replaced by a general polytope and scalable frames lose their barycentric interpretation. Which geometric configuration plays the role of the orthocentric simplex for non-minimal $s$-tight IC measurements?
    \item \emph{Realisability.} Given a GGPT with its canonical Euclidean structure -- in the quantum case, the Hilbert--Schmidt one -- which acute orthocentric simplices can actually be anchored in its state space? In the quantum case this asks for a characterisation of the admissible probability vectors $p$ for a given dimension $D$; the maximally symmetric point of this problem touches upon the existence of SICs.
    \item \emph{Primary equation.} The generalised Urgleichung is by now available for the whole $s$-tight class~\cite{szymusiak2025morphophoricity}, extending the morphophoric case~\cite{szymusiak2025can}. 
    What remains open is the accompanying theory of state update: Bayes' rule of Appendix~\ref{app:instrument} constrains the post-measurement states, but a characterisation of the instruments compatible with a given $s$-tight IC measurement -- including their sequential statistics and repeatability properties -- is still lacking; first steps in this direction are taken in the companion work~\cite{slomczynski2026fiducial}.
   
\end{enumerate}

We hope that the bridge established here -- between the operational demand for optimal information extraction and the concurrence of altitudes -- will prove useful on both of its ends.

\section*{Acknowledgements}

We are grateful to Karol Życzkowski, whose question prompted us to complement the abstract theory with the explicit quantum example of Section~\ref{sec:axial}.

This work was supported by the European Research Council (ERC) Advanced Grant ``Typical and Atypical structures in quantum theory'' (TAtypic) under the European Union's Horizon Europe research and innovation programme (grant agreement No.~101142236).

\appendix
\section{Instrument and Bayes' rule} \label{app:instrument}

Let $x = \lambda x_1 + (1-\lambda)x_2 \in S$ be a convex combination (statistical mixture) of two states, with $x_1, x_2 \in S$, $\lambda \in (0,1)$, and let $j=1,\dots, n$.
As $\pi_j$ is affine, we have
\[
  \pi_j(x) = \lambda \pi_j(x_1) + (1-\lambda)\pi_j(x_2).
\]
Moreover, the requirement that $\mathcal{I}_j = \pi_j F_j$ is affine implies
\[
  \pi_j(x)\, F_j(x)
  = \lambda\, \pi_j(x_1)\, F_j(x_1) + (1-\lambda)\, \pi_j(x_2)\, F_j(x_2).
\]
Dividing both sides by $\pi_j(x)$ (assuming $\pi_j(x) \neq 0$) and substituting the affine expansion of $\pi_j(x)$ into the denominator, we get 
\[ 
  F_j(x) = \left( \frac{\lambda \pi_j(x_1)}{\lambda \pi_j(x_1) + (1-\lambda) \pi_j(x_2)} \right) F_j(x_1) + \left( \frac{(1-\lambda) \pi_j(x_2)}{\lambda \pi_j(x_1) + (1-\lambda) \pi_j(x_2)} \right) F_j(x_2). 
\]

This has a direct Bayesian interpretation. 
Consider a classical two-hypothesis inference problem: the system was prepared in state $x_1$ with \textit{prior probability} $\lambda$, or in state $x_2$ with \textit{prior probability} $1-\lambda$. 
Upon observing outcome $j$, Bayes' rule gives the \textit{posterior probability} that the system was prepared in state $x_i$:
\[
  \Pr(\text{prepared in } x_i \mid \text{outcome is } j)
  = \frac{\lambda_i\, \pi_j(x_i)}{\pi_j(x)},
\]
where $\lambda_1 = \lambda$, $\lambda_2 = 1-\lambda$, and $\pi_j(x_i)$ is the \textit{likelihood} of outcome $j$ given preparation $x_i$. 
The post-measurement state $F_j(x)$ is therefore the convex combination of $F_j(x_1)$ and $F_j(x_2)$ weighted by these posteriors.
In other words, the affinity of $\mathcal{I}_j = \pi_j F_j$ is the mathematical expression of the requirement that state update is consistent with Bayesian inference about the preparation.

\section{Equivalence of approaches} \label{app:equivalence}

In Section \ref{sec:newapproach} we showed a path from the GGPT $(V,C,B,\langle\cdot,\cdot\rangle_0,m,\mu)$ to the simpler representation given by $(V_0,\langle\cdot,\cdot\rangle_0,S)$ that preserves the \emph{internal} geometry of the theory. Now we will show that this triple alone allows for the reconstruction of the GGPT  up to the choice of the parameter $\mu$, which now characterises the \emph{external} geometry of the theory.

We begin with a finite-dimensional real vector space $V_0$ equipped with an inner product $\langle\cdot,\cdot \rangle_0$ and a compact convex neighbourhood of the origin $S\subset V_0$ to build a GPT. States will be now represented by the elements of $S$. In order to build the set of effects we start with defining an ancillary vector space $V_0\times \mathbb R$ with nonstandard vector space operations by introducing a bijection $F \colon V_0\times \mathbb R \to V_0\times \mathbb R$ given by
\[
   F(\psi,q) \coloneqq \begin{cases} (q\psi, q) & \text{if } q \neq 0, \\ (\psi, 0) & \text{if } q = 0. \end{cases}
\]
We then equip $V_0\times \mathbb R$ with the linear structure pulled back from the canonical vector space operations on $V_0\times \mathbb R$ via $F$, i.e., $x + y \coloneqq F^{-1}(F(x) + F(y))$ and $t \cdot x \coloneqq F^{-1}(t F(x))$. This guarantees that $V_0\times \mathbb R$ is a vector space and, in particular, the linear structure on $V_0\times\{0\}$ coincides with that on $V_0$. Let $S^\star \coloneqq \{y\in V_0:\langle y,x\rangle\geq -1 \quad \text{for all }x\in S\}$. Then the set of effects $\mathcal E$ is defined as 
\[
   \mathcal E \coloneqq \left\{(\psi,c)\in S^\star\times(0,1):\frac{c}{c-1}\psi\in S^\star\right\}\cup\{(0,1),(0,0)\}.
\]
To show that $\mathcal E$ is convex, let us take $t \in [0,1]$ and first observe that for $(\psi_1,c_1),(\psi_2,c_2)\in\mathcal E$ with $c_1,c_2 \in (0,1)$ and for any $y\in S$ we get
\[
   \left\langle \frac{tc_1\psi_1+(1-t)c_2\psi_2}{tc_1+(1-t)c_2},y \right\rangle=\frac{tc_1\langle\psi_1,y\rangle+(1-t)c_2\langle\psi_2,y\rangle}{tc_1+(1-t)c_2}\geq \frac{-tc_1-(1-t)c_2}{tc_1+(1-t)c_2}=-1,
\]
which shows that $\frac{tc_1\psi_1+(1-t)c_2\psi_2}{tc_1+(1-t)c_2}\in S^\star$.
Now, since the condition $\frac{c}{c-1}\psi\in S^\star$ can be equivalently expressed as $c\langle\psi,y\rangle\leq 1-c$ for all $y\in S$, we also get
\begin{align*}
   (tc_1+(1-t)c_2)\left\langle \frac{tc_1\psi_1+(1-t)c_2\psi_2}{tc_1+(1-t)c_2},y \right\rangle &=tc_1\langle\psi_1,y\rangle+(1-t)c_2\langle\psi_2,y\rangle \\
   &\leq t(1-c_1)+ (1-t)(1-c_2)\\
   &= 1-(tc_1+(1-t)c_2)
\end{align*}
and thus $t(\psi_1,c_1)+(1-t)(\psi_2,c_2)=\left(\frac{tc_1\psi_1+(1-t)c_2\psi_2}{tc_1+(1-t)c_2},tc_1+(1-t)c_2\right)\in \mathcal E$. 
The convex combinations involving the boundary elements $(0,0)$ and $(0,1)$ are handled by a direct check; e.g., $t(\psi_1,c_1)+(1-t)(0,0)=(\psi_1,tc_1)$, and the condition $\frac{tc_1}{tc_1-1}\psi_1\in S^\star$ follows from the convexity of $S^\star$ and $0\in S^\star$, since $\frac{tc_1}{1-tc_1}\leq\frac{c_1}{1-c_1}$.
Finally, the function assigning probability takes the form 
\[
   \mathfrak{p}(x,(\psi,c))=c(\langle\psi,x\rangle+1).
\] 
It is easy to see that $0\leq \mathfrak{p}(x,(\psi,c))\leq 1$ for every $x\in S$ and $(\psi,c)\in\mathcal E$, and that $\mathfrak p$ is affine in the state and linear in the effect (with respect to the linear structure introduced above). Note that, if we measure the effect $(\psi,c)$, then $c$ can  be interpreted as the probability of obtaining the answer ``yes'' if the state of the system was initially $0$.

To see the equivalence of this approach, we will now reconstruct the GGPT $(V,C,B,\langle\cdot,\cdot\rangle_0,m,\mu)$. First, we build a linear space $V$ as an abstract orthogonal direct sum of $V_0$ and some 1-dimensional vector space $V_1$, i.e., $V \coloneqq V_0\oplus V_1$. We choose $m\in V_1\setminus\{0\}$ and set $\mu \coloneqq \lVert m \rVert^2$. Note that, due to the one-dimensionality of $V_1$, this choice involves only one free parameter: $\mu$. So we will obtain a 1-parameter family of GGPTs, all of which share the same internal geometry. We define the set of states $B \coloneqq \sqrt{\mu}S+m$ and, consequently, the convex cone $C \coloneqq \{\alpha x:\alpha\geq 0,x\in B\}$. Note that $C$ does not depend on the choice of $m$ (and thus $\mu$). Indeed, for $m,m'\in V_1\setminus\{0\}$ with squared lengths $\mu$ and $\mu'$ we necessarily have $m'=\sqrt{\frac{\mu'}{\mu}}m$ and thus $B'=\sqrt{\mu'}S+m'=\sqrt{\frac{\mu'}{\mu}}B$. Having defined $V,C$ and $B$, we also obtain $V^*,C^*$ and a unique $e\in C^*$ such that $B$ lies in the $1$-level set of $e$; note that
$e=T(\tfrac{1}{\mu}m)$, since
$\langle\tfrac{1}{\mu}m,\sqrt{\mu}y+m\rangle=1$ for every $y\in S$. Since $V$ is now from the beginning equipped with an inner product, we also get an induced isometric linear isomorphism $T:V\to V^*$. We can now establish the one-to-one correspondence between the sets of effects via $(\psi,c)\mapsto T(\tfrac{c}{\sqrt{\mu}}\psi+\frac{c}{\mu}m)$. In order to see that this function maps effects onto effects, let us first observe that 
\begin{align*}
    T(\frac{c}{\sqrt{\mu}}\psi+\frac{c}{\mu}m)\in C^*&\iff \left\langle\frac{c}{\sqrt{\mu}}\psi+\frac{c}{\mu}m,x\right\rangle\geq 0 \quad \text{for all }x\in B\\
    &\iff \left\langle\frac{c}{\sqrt{\mu}}\psi+\frac{c}{\mu}m,\sqrt{\mu}y+m\right\rangle\geq 0 \quad \text{for all }y\in S\\
     &\iff c(\langle\psi,y\rangle+1)\geq 0 \quad \text{for all }y\in S\\
    &\iff \psi\in S^\star.
\end{align*}

Analogously, since $e-T(\tfrac{c}{\sqrt{\mu}}\psi+\tfrac{c}{\mu}m) = T\bigl(\tfrac{1}{\mu}m-\tfrac{c}{\sqrt{\mu}}\psi-\tfrac{c}{\mu}m\bigr)$, we obtain
\begin{align*}
    e - T(\tfrac{c}{\sqrt{\mu}}\psi+\tfrac{c}{\mu}m)\in C^*
    &\iff (1-c) - c\langle\psi,y\rangle \geq 0 \quad \text{for all }y\in S\\
    &\iff \psi = 0 \textnormal{ (if } c=1\textnormal{)} \quad\textnormal{or}\quad \tfrac{c}{c-1}\psi\in S^\star \textnormal{ (if } c<1\textnormal{)}.
\end{align*}
Together, these two equivalences show that $(\psi,c)\mapsto T(\frac{c}{\sqrt{\mu}}\psi+\frac{c}{\mu}m)$ is a bijection between $\mathcal E$ and the set of effects of the reconstructed GGPT, which completes the proof of the equivalence of the two approaches.

\end{document}